\documentclass[a4paper,11pt]{amsart}

\usepackage{a4wide, amsmath, amsfonts, amssymb, mathrsfs, amsthm, bbm, color, stmaryrd, breqn}
\usepackage[hyperfootnotes=false]{hyperref}
\numberwithin{equation}{section}

\newtheorem{theorem}{Theorem}[section]
\newtheorem{lemma}[theorem]{Lemma}
\newtheorem{corollary}[theorem]{Corollary}
\newtheorem{remark}[theorem]{Remark}
\newtheorem{proposition}[theorem]{Proposition}
\newtheorem{definition}[theorem]{Definition}
\newtheorem{example}[theorem]{Example}

\renewcommand{\d}{\,\mathrm{d}}
\newcommand{\R}{\mathbb{R}}
\newcommand{\F}{\mathcal{F}}
\newcommand{\N}{\mathbb{N}}
\newcommand{\Z}{\mathbb{Z}}
\newcommand{\E}{\mathbb{E}}
\newcommand{\1}{\mathbf{1}}
\renewcommand{\P}{\mathbb{P}}

\newcommand{\Hc}{\mathcal{H}}
\newcommand{\oP}{\overline{P}}
\newcommand{\oQ}{\overline{Q}}
\newcommand{\argmin}{\mathrm{argmin}}

\newcommand{\U}{\mathrm{U}}
\newcommand{\D}{\mathrm{D}}
\newcommand{\G}{\mathcal{G}}

\title{A superhedging approach to stochastic integration}

\author[\L ochowski]{Rafa\l{} M. {}\L ochowski}
\address{Rafa\l{} M. \L ochowski, Warsaw School of Economics, Poland, and African Institute for Mathematical Sciences, South Africa}
\email{rlocho314@gmail.com}

\author[Perkowski]{Nicolas Perkowski}
\address{Nicolas Perkowski, Humboldt-Universit\"at zu Berlin, Germany}
\email{perkowsk@math.hu-berlin.de}

\author[Pr\"omel]{David J. Pr\"omel}
\address{David J. Pr\"omel, University of Oxford, United Kingdom}
\email{proemel@maths.ox.ac.uk}

\date{\today}

\begin{document}

\begin{abstract}
  Using Vovk's outer measure, which corresponds to a minimal superhedging price, the existence of quadratic variation is shown for ``typical price paths'' in the space of c\`adl\`ag functions possessing a mild restriction on the jumps directed downwards. In particular, this result includes the existence of quadratic variation of ``typical price paths'' in the space of non-negative c\`adl\`ag paths and implies the existence of quadratic variation in the sense of F\"ollmer quasi surely under all martingale measures. Based on the robust existence of the quadratic variation, a model-free It\^o integration is developed.
\end{abstract}

\maketitle


% Explanation
% -----------
% 60H05 ~ Stochastic Analysis (Stochastic integrals)
% 91G99 ~ Mathematical finance (None of the above, but in this section)

\section{Introduction}
{\let\thefootnote\relax\footnotetext{The present work is a replacement of the technical report~\cite{Lochowski2015} and an unpublished work the second and third author.}}

In a recent series of papers~\cite{Vovk2009,Vovk2012,Vovk2015}, Vovk introduced a model-free, hedging-based approach to mathematical finance that uses arbitrage considerations to examine which path properties are satisfied by ``typical price paths''. For this purpose an outer measure~$\overline{P}$ is introduced, which corresponds to a superhedging price, and a statement~(A) is said to hold for ``typical price paths'' if (A) is true except on a null set under $\overline{P}$. We also refer to~\cite{Schafer2001} and~\cite{Takeuchi2009} for related settings. As a nice consequence all results proven for typical price paths hold quasi surely under all martingale measures and even more generally quasi surely under all semimartingale measures for which the coordinate process satisfies the classical condition of ``no arbitrage of the first kind'', also known as ``no unbounded profit with bounded risk''.  

The path properties of typical price paths belonging to the space of continuous functions are already rather well studied. Vovk proved that the path regularity of typical continuous price paths is similar to that of (continuous) semimartingales (see~\cite{Vovk2008}), i.e. non-constant typical continuous price paths have infinite $p$-variation for $p<2$ but finite $p$-variation for $p>2$, and they possess a quadratic variation (see~\cite{Vovk2012}). More advanced path properties such as the existence of an associated local time or an It\^o rough path were shown in~\cite{Perkowski2015,Perkowski2016}. All these results give a robust justification for taking additional path properties as an underlying assumption in model-independent finance or mathematical finance under Knightian uncertainty. 

However, while in financial modeling stochastic processes allowing for jumps play a central role, typical price paths belonging to the space $D([0,T],\R^d)$ of c\`adl\`ag functions are not well understood yet. This turns out to be a more delicate business because of two reasons in particular. First there exists no canonical extension of Vovk's outer measure~$\overline{P}$ to the whole space $D([0,T],\R^d)$ as discussed in Remark~\ref{rmk:outer measure} and thus we need to work with suitable subspaces. Second, the class of admissible strategies gets smaller if the class of possible price trajectories gets bigger.

\smallskip
The main aim of the present article is to develop ``model-free'' It\^o integration of adapted c\`agl\`ad (and even more general) integrands with respect to typical price paths. As the classical It\^o integral is one of the key stones of stochastic calculus and  mathematical finance in continuous time, the model-free It\^o integral constitutes a step towards a ``model-free'' stochastic calculus and is potentially a useful tool in model-independent finance.

For this purpose we consider the underlying space~$\Omega_{\psi}$ given by any subset of c\`adl\`ag functions with mildly restricted jumps directed downwards, that is every $\omega=(\omega^1,\dots,\omega^d) \in \Omega_\psi$ satisfies
\begin{equation*}
  \omega^i (t-) - \omega^i (t)\leq \psi \bigg(\sup_{s\in[0,t)}|\omega(s)|\bigg),\quad t\in (0,T],
\end{equation*}
where $\omega^i (t-):= \lim_{s\to t, \, s<t} \omega (s)$ for $i=1,\dots,d$ and $\psi\colon\R_+\to\R_+$ is a fixed non-decreasing function. We would like to remark that the sample space~$\Omega_\psi$ covers several sample spaces previously considered in the related literature such as the space of
continuous functions or the space of non-negative c\`adl\`ag functions, see Example~\ref{ex:sample spaces}. 

\smallskip
Our first contribution is to prove that the quadratic variation exists for typical price paths, and that it is given as the uniform limit of discrete versions of the quadratic variation. For the precise result we refer to Theorem~\ref{thm:quadratic variation} and Corollary~\ref{cor:quadratic variation}. Intuitively, this means that it is possible to get infinitely rich by investing in those paths where the convergence of the discrete quadratic variations fails. Let us emphasize once again that this, in particular, justifies the assumption of the existence of the quadratic variation in model-free financial mathematics. In the case of continuous paths or c\`adl\`ag paths with restricted jumps (in all directions), the existence of the quadratic variation for typical price paths has been shown by Vovk in~\cite{Vovk2012} and~\cite{Vovk2015}, respectively. However, the case of, for example, non-negative c\`adl\`ag paths stayed open since here no short-selling is allowed and therefore the arguments used in~\cite{Vovk2012, Vovk2015} break down.

The existence of the quadratic variation is not only a crucial ingredient to develop a model-free It\^o integration, it is also a powerful assumption by itself in robust and model-independent financial mathematics since it, in particular, allows to use F\"ollmer's pathwise It\^o integral~\cite{Follmer1979} and thus to construct $\int f^\prime (\omega(s))\d \omega(s)$ for $f\in C^{2}$ or more generally for path-dependent functionals in the sense of Dupire~\cite{Dupire2009} as shown by Cont and Fourni\'e~\cite{Cont2010a}. The F\"ollmer integration has a long history of successful applications in mathematical finance going back at least to Bick and Willinger~\cite{Bick1994} and Lyons~\cite{Lyons1995}. Moreover, in many recent works in robust financial mathematics the existence of quadratic variation serves as a basic assumption on the underlying price paths, see for instance~\cite{Davis2014,Schied2016,Riga2016,Cuchiero2016,Beiglbock2017}.

\smallskip
Our second contribution is the development of a model-free It\^o integration theory for adapted c\`agl\`ad integrands with respect to typical price paths, see Theorem~\ref{thm:integral}. Compared to the rich list of classical pathwise constructions of stochastic integrals, e.g. \cite{Bichteler1981}, \cite{Willinger1988}, \cite{Willinger1989}, \cite{Karandikar1995}, \cite{Nutz2012}, the presented construction complements the previous works in three aspects: 
\begin{enumerate}
  \item It works without any tools from probability theory at all and is entirely based on pathwise superhedging arguments.
  \item The non-existence of the It\^o integral comes with a natural arbitrage interpretation, that is, it is possible to achieve a pathwise arbitrage opportunity of the first kind if the It\^o integral does not exist.
  \item The model-free It\^o integral possesses continuity estimates (in a topology induced by Vovk's outer measure~$\overline{P}$).
\end{enumerate}

Let us emphasize that the continuity of model-free integral is one of most important aspects since already in the classical probabilistic setting most applications of the It\^o integral (SDEs, stochastic optimization, duality theory,...) are based on the fact that it is a continuous operator. Furthermore, we would like to point out that Vovk~\cite{Vovk2016} very recently introduces a related approach to define a model-free It\^o integral, which satisfies also  (1) and (2) but provides no continuity estimates. We comment on the difference between our construction and Vovk's  work in more detail in Remark \ref{q_var_part_ind}.

\smallskip
The construction of our model-free It\^o integral  crucially builds on the existence of the quadratic variation for typical price paths and on continuity estimates for pathwise integrals of step functions. We distinguish between the general sample space $\Omega_\psi$ and its restriction to the space of continuous paths since the latter  leads to significantly better continuity estimates.

In the case of continuous paths the continuity estimate for the pathwise It\^o integral of a step function, see Lemma~\ref{lem:model free ito cont}, relies on Vovk's pathwise Hoeffding inequality (\cite[Theorem~A.1]{Vovk2012}). In this context, we recover essentially the results of~\cite[Theorem~3.5]{Perkowski2016}. However, in~\cite{Perkowski2016} we worked with the uniform topology on the space of integrands while here we are able to strengthen our results and to replace the uniform distance with a more natural distance that depends only on the integral of the squared integrand against the quadratic variation. 

In the general case of c\`adl\`ag paths the continuity estimates for the integral of step functions, see Lemma~\ref{lem: ito inequality cadlag}, require completely different techniques compared to the continuous case. In particular, while one has a very precise control of the fluctuations in the case of continuous price paths, this is not possible anymore in the presence of jumps as c\`adl\`ag  paths could have a ``big'' jump at any time. Hence, Vovk's pathwise Hoeffding inequality needs to be replaced by a pathwise version of the Burkholder-Davis-Gundy inequality due to Beiglb\"ock and Siorpaes~\cite{Beiglbock2015}.

\subsection*{Organization of the paper}
Section~\ref{sec:vovk} introduces Vovk's model-free and hedging-based approach to mathematical finance. In Section~\ref{sec:quadratic variation} the existence of quadratic variation for typical (c\`adl\`ag) price paths is shown. The model-free It\^o integration is developed in Section~\ref{sec:ito integration}. Appendices~\ref{sec:appendix} and~\ref{app:stopping} collect some auxiliary results concerning Vovk's outer measure and stopping times. 

\subsection*{Acknowledgment}
The research of R.M.\L.\ was partially founded by the National Science Centre, Poland, under Grant No.~$2016/21/$B/ST$1/01489$ and the African Institute for Mathematical Sciences, South Africa. N.P. and D.J.P. are grateful for the excellent hospitality of the Hausdorff Research Institute for Mathematics (HIM), where the work was partly initiated. D.J.P.
was employed at ETH Z\"urich when the main part of the work was elaborated and gratefully acknowledges the financial support of the Swiss National Foundation under Grant No.~$200021\_163014$.

\section{Superhedging and typical price paths}\label{sec:vovk}

Vovk's model-free and hedging-based approach to mathematical finance allows for determining sample path properties of ``typical price paths''. For this purpose he introduces a notion of outer measure which is based on purely pathwise arbitrage considerations, see for example~\cite{Vovk2012}. Following a slightly modified framework as introduced in~\cite{Perkowski2015,Perkowski2016}, we briefly set up the notation and definitions.\smallskip

For a positive integer~$d$ and a finite time horizon $T\in (0,\infty)$ let $D([0,T],\R^d)$ be the space of all c\`adl\`ag functions $\omega\colon [0,T]\to\R^d$. For $t \in (0, T]$ let us define $\omega^i (t-):= \lim_{s\to t, \, s<t} \omega^i (s)$ and by $\R_+$ let us denote the interval $[0, \infty)$. For a fixed non-decreasing function $\psi\colon\R_+\to\R_+$ and a fixed set $\Omega \subseteq D([0,T],\R^d)$ we consider the sample space~$\Omega_\psi$ of possible price paths given by 
\begin{equation*}
  \Omega_\psi := \left \{ \omega=(\omega^1,\dots, \omega^d) \in \Omega \,:\, \omega^i(t-)-\omega^i (t)\leq \psi\bigg(\sup_{s\in[0,t)}|\omega(s)|\bigg)\,\forall t\in (0,T],\,i=1,\dots, d \right \}.
\end{equation*}
Notice that the function~$\psi$ in the definition of $\Omega_\psi$ gives a predictable bound for the allowed jump size for jumps directed downwards (for all price paths belonging to~$\Omega_\psi$). However, the jumps directed upwards are not necessarily restricted in any way. This general type of underlying sample spaces provides a unifying setting for many examples of previously treated sample spaces in the related literature. 

\begin{example}\label{ex:sample spaces}
  The following sample spaces are examples of~$\Omega_\psi$:
  \begin{enumerate}
    \item $\Omega_c:=C([0,T],\R^d)$, the space of all continuous functions $\omega\colon [0,T]\to\R^d$, 
    \item $\Omega_+ := D([0,T], \R^d_+)$, the space of all non-negative c\`adl\`ag functions $\omega\colon [0,T]\to\R^d_+$, 
    \item $\tilde \Omega_{\psi}$ which is defined as the subset of all c\`adl\`ag functions $\omega\colon [0,T]\to\R^d$ such that
          \begin{equation*}
            |\omega (t)-\omega(t-)|\leq \psi\bigg(\sup_{s\in[0,t)}|\omega(s)|\bigg),\quad t\in (0,T],
          \end{equation*}
          and $\psi\colon\R_+\to (0,\infty)$ is a fixed non-decreasing function.
   \end{enumerate}
   A detailed discussion about the financial interpretation of the last space can be found in~\cite{Vovk2015} and a generalization of this space allowing for different bounds for jumps directed upwards resp. downwards was recently introduced in~\cite{Vovk2016}.
\end{example}

The coordinate process on $\Omega_\psi$ is denoted by $S$, i.e. $S_{t}(\omega):=\omega(t)$ for $\omega \in \Omega_\psi$ and $t\in [0,T]$. For each $t\in [0,T]$, ${\mathcal{F}}_{t}^{\circ}$ is defined to be the $\sigma$-algebra on~$\Omega_\psi$ that is generated by $(S_s: s \in [0,t])$ and ${\mathcal{F}}_{t}$ is the universal completion of ${\mathcal{F}}_{t}^{\circ}$.\footnote{The reason for working with the universal completion is that this provides us with many useful stopping times, see Appendix~\ref{app:stopping} for details.} An event is an element of the $\sigma$-algebra ${\mathcal{F}}_{T}.$ Stopping times $\tau\colon\Omega_\psi \to [0,T]\cup \{ \infty \} $ with respect to the filtration $({\mathcal{F}}_{t})_{t\in[0,T]}$ and the corresponding $\sigma$-algebras ${\mathcal{F}}_{\tau}$ are defined as usual. The indicator function of a set $A$, for $A\subseteq \R^d$ or $A\subseteq \Omega_\psi$, is denoted by $\1_A$ and for two real vectors $x,y\in \R^d$ we write $xy=x\cdot y$ for the inner product on $\R^d$, and $\lvert\cdot\rvert$ always denotes the $\ell^2$-norm on $\R^d$. For $\omega\in D([0,T],\R^d)$ the supremum norm is given by $\|\omega\|_{\infty}:= \sup_{t\in [0,T]} |\omega(t)|$. Furthermore, we use the notation $s\wedge t:= \min \{s,t\}$, $s\vee t:= \max\{s,t\}$ for $s,t\in\R_+$, $\N:=\{1,2,\dots\}$ for the set of positive integers, $\N_0 := \N \cup \{0\}$ and $\Z$ for set of all integers. \smallskip

A process $H\colon \Omega_\psi\times [0,T]\to\R^d$ is a \emph{simple (trading) strategy} if there exist a sequence of stopping times $0 = \tau_0 < \tau_1 <  \tau_2 < \dots$, and $\F_{\tau_n}$-measurable bounded functions $h_n\colon \Omega_\psi \to \R^d$, such that for every $\omega\in\Omega_\psi$, $\tau_{n}(\omega)=\tau_{n+1}(\omega)=\ldots\in [0,\infty]$ from some $n=n(\omega)\in \N$ on, and such that 
\begin{equation*}
  H_t(\omega) = \sum_{n=0}^\infty h_n(\omega) \1_{(\tau_n(\omega),\tau_{n+1}(\omega)]}(t),\quad t \in [0,T].
\end{equation*}
Therefore, for a simple strategy $H$ the corresponding integral process
\begin{equation*}
  (H \cdot S)_t(\omega) := \sum_{n=0}^\infty h_n(\omega) \cdot (S_{\tau_{n+1} \wedge t}(\omega) - S_{\tau_n\wedge t}(\omega)) = \sum_{n=0}^\infty h_n(\omega) S_{\tau_n\wedge t, \tau_{n+1} \wedge t}(\omega) 
\end{equation*}
is well-defined for all $\omega \in \Omega_\psi$ and all $t \in [0,T]$; here we introduced the notation $S_{u,v}:= S_v - S_u$ for $u,v \in [0,T]$.

For $\lambda > 0$ a simple strategy $H$ is called \emph{(strongly) $\lambda$-admissible} if $(H\cdot S)_t(\omega) \ge - \lambda $ for all $(t,\omega) \in [0,T]\times \Omega_\psi$. The set of strongly $\lambda$-admissible simple strategies is denoted by $\mathcal{H}_{\lambda}$.\smallskip

In the next definition we introduce an outer measure $\oP$, which is very similar to the one used by Vovk~\cite{Vovk2012}, but not quite the same. We refer to \cite[Section~2.3]{Perkowski2016} for a detailed discussion of the relation between our slightly modified outer measure and the original one due to Vovk.

\begin{definition}\label{def:vovk}
  \emph{Vovk's outer measure} $\overline{P}$ of a set $A \subseteq \Omega_\psi$ is defined as the minimal superhedging price for $\1_A$, that is
  \begin{equation*}
    \overline{P}(A) := \inf\Big\{\lambda > 0\,: \,\exists (H^n)_{n\in \N} \subset \mathcal{H}_{\lambda} \text{ s.t. } \forall \omega \in \Omega_\psi\,: \, \liminf_{n \to\infty} (\lambda + (H^n\cdot S)_T(\omega)) \ge \1_A(\omega)\ \Big\}.
  \end{equation*}
  A set $A \subseteq \Omega_\psi$ is called a \emph{null set} if it has outer measure zero. A property (P) holds for \emph{typical price paths} if the set $A$ where (P) is violated is a null set.
\end{definition}  

Indeed, as in \cite[Lemma~4.1]{Vovk2012} or \cite[Lemma~2.3]{Perkowski2016}, it is straightforward to verify that~$\overline{P}$ fulfills all properties of an outer measure. 

\begin{lemma}
  $\overline{P}$ is an outer measure, i.e. a non-negative set function defined on the subsets of $\Omega_\psi$ such that $\overline{P}(\emptyset) = 0$, $\overline{P}(A) \leq \overline{P}(B)$ for $A \subseteq B \subseteq \Omega_\psi$, and $\overline{P}(\cup_n A_n) \leq \sum_n \overline{P}(A_n)$ for every sequence of subsets $(A_n)_{n\in\N}\subseteq \Omega_\psi$.
\end{lemma}

\begin{remark}\label{rmk:outer measure}
  The freedom of choosing the set~$\Omega$ in the definition of $\Omega_\psi$ is very much in the spirit of so-called ``sets of beliefs'' or ``prediction sets'' recently introduced by Hou and Ob{\l}{\'o}j~\cite{Hou2015} in the context of model-free mathematical finance. Their idea is to allow for choosing an a priori prediction set~$\tilde \Omega \subseteq D([0,T],\R^d)$ and to require ``superhedging'' only on~$\tilde \Omega$, which has, in particular, the desirable effect to lead to lower minimal pathwise superhedging prices.

  The same effect can be observed for Vovk's outer measure~$\overline{P}$ as it corresponds to a minimal superhedging price on $\Omega_\psi$, cf. Definition~\ref{def:vovk}. For example, the a priori restrictions of the general sample space $\Omega_\psi$ to the specific choice $\Omega_c$ lead to stronger estimates of model-free It\^o integrals, cf. Subsection~\ref{subsec:integration continuous} and~\ref{subsec:integration cadlag}. However, when showing results for typical price paths, one wants to take the underlying sample space as big as possible since any result (for typical price paths) immediately transfers to any smaller ``prediction set''.

  With this observation in mind, it might be desirable to consider even the more general sample space $\Omega_\psi=D([0,T],\R^d)$ of all c\`adl\`ag functions as possible price paths. Unfortunately, there seems to be (currently) no canonical extension of Vovk's outer measure to $\Omega_\psi=D([0,T],\R^d)$ because the condition of strong $\lambda$-admissibility would only allow for the trivial trading strategy $H \equiv 0$. Therefore, it is necessary to include at least some predictable bound~$\psi$ on the jumps directed downwards.
\end{remark}

One of the main reasons why Vovk's outer measure~$\overline{P}$ is of such great interest in model-independent financial mathematics, is that it dominates simultaneously all local martingale measures on the sample space~$\Omega_\psi$ (cf. \cite[Lemma~6.3]{Vovk2012} and \cite[Proposition~2.6]{Perkowski2016}). In other words, a null set under~$\overline{P}$ turns out to be a null set simultaneously under all local martingale measures on~$\Omega_\psi$.

\begin{proposition}\label{prop:local martingale}
  Let $\P$ be a probability measure on $(\Omega_\psi, \F)$, such that the coordinate process~$S$ is a $\P$-local martingale, and let $A \in \F_T$. Then $\P(A) \leq \overline{P}(A)$.
\end{proposition}

A second reason is that sets with outer measure zero come with a natural arbitrage interpretation from classical mathematical finance. Roughly speaking, a null set leads to a pathwise arbitrage opportunity of the first kind (NA1), see \cite[p.~564]{Vovk2012} and \cite[Lemma~2.4]{Perkowski2016}.

\begin{proposition}\label{prop:na1 interpretation}
  A set $A \subseteq \Omega_\psi$ is a null set if and only if there exists a constant $K\in (0,\infty)$ and a sequence of strongly $K$-admissible simple strategies $(H^n)_{n\in \N} \subset \mathcal{H}_K$ such that
  \begin{equation}\label{eq:NA1 with simple strategies}
    \liminf_{n \to \infty} (K + (H^n \cdot S)_T(\omega )) \ge \infty \cdot \1_A(\omega),\quad \omega \in \Omega_\psi,
  \end{equation}
  with the convention $ \infty\cdot 0 := 0$ and $ \infty \cdot 1 := \infty$. In that case we can take $K > 0$ arbitrarily small.
\end{proposition}

Since the proofs of Proposition~\ref{prop:local martingale} and Proposition~\ref{prop:na1 interpretation} work exactly as in the case of continuous price paths, they are postponed to  Appendix~\ref{sec:appendix}.

\subsection{Auxiliary set function}

In order to prove the existence of the quadratic variation for typical price paths belonging to the sample space $\Omega_\psi$, we need to introduce a relaxed version of admissibility.\smallskip

A simple strategy $H$ is called \emph{weakly $\lambda$-admissible} if for all $(t,\omega)\in [0,T]\times\Omega_\psi$
\begin{equation*}
  (H\cdot S)_t(\omega) \geq - \lambda\big(1+ | S_{\rho_\lambda(H)}(\omega)|\1_{[\rho_\lambda(H)(\omega), T]}(t)\big),
\end{equation*}
where                                                                                                                                                                                                                                                                                                                                                                                                                                                                                                                                                                                                                                                                                                                                                                          
\begin{equation*}
  \rho_\lambda(H)(\omega):=\inf\big\{t \in [0,T]\,:\, (H\cdot S)_t(\omega) \leq -\lambda\big\}
\end{equation*}
and 
\begin{equation*}
  H_t(\omega) = H_t(\omega) \1_{[0,\rho_\lambda(H)(\omega)\wedge T]}(t).
\end{equation*}
In all definitions we apply the conventions $\inf \emptyset := \infty$, $[\infty, T]:=\emptyset$ and $S_\infty (\omega):=0$. By Lemma~\ref{lem:stopping} in Appendix \ref{app:stopping} $\rho_\lambda(H)$ is a stopping time. We write $\G_\lambda$ for the set of weakly $\lambda$-admissible strategies. Expressed verbally, this means that weakly $\lambda$-admissible strategies must give a payoff larger than $-\lambda$ at all times, except that they can lose all their previous gains plus $\lambda (1 + |S_t |)$ through one large jump; however, in that case they must instantly stop trading and may not try to bounce up again. \smallskip

Based on the notion of weak admissibility, we define an auxiliary set function $\overline{Q}$ via the minimal superhedging price for a new class of trading strategies.

\begin{definition}\label{def:set function}
  The set function $\overline{Q}$ is given by
  \begin{align*}
    \overline{Q}(A) := \inf\Big\{\lambda > 0\,:&\,\exists (H^n)_{n\in \N} \subset \mathcal{G}_{\lambda} \text{ s.t. }\forall\omega \in \Omega_\psi\,: \\
    &\liminf_{n \to \infty} \big(\lambda + (H^n\cdot S)_T(\omega) + \lambda \mathbbm{1}\cdot S_{\rho_\lambda(H^n)}(\omega) \1_{\{\rho_\lambda(H^n)<\infty\}}(\omega)\big) \ge \1_A(\omega) \,  \Big\},
  \end{align*}
  where $\mathbbm{1}:=(1,\dots,1)\in\R^d$. %and $\{\rho_\lambda(H^n)<\infty\}:=\left\{\omega \in \Omega_\psi \,:\, \rho_\lambda(H^n)(\omega)<\infty\right\}$.
\end{definition}

\begin{remark}
  Note that $\overline{Q}$ is not an outer measure since it fails to be even finitely subadditive. However, it will be a useful tool to show the existence of the quadratic variation for typical price paths in~$\Omega_\psi$ and for the construction of model-free It\^o integration. For events containing only uniformly bounded paths, the following lemma (Lemma~\ref{lem:outer measure equivalence}) shows that it is \emph{nearly} subadditive:
  \begin{equation*}
    \oQ(\cup_n A_n) \leq \oP(\cup_n A_n) \leq \sum_{n=1}^\infty \oP(A_n) \leq (1+ 3 d K +2 d \psi (K)) \sum_{n=1}^\infty \oQ(A_n),
  \end{equation*}
  for every sequence $(A_n)_{n\in \N}$ such that $(A_n)_{n\in \N}\subseteq \left \{ \omega\in \Omega_\psi \,:\, \| \omega\|_\infty \leq K\right \}$ for some constant $K\in \R_+$. Hence, in the particular case where $\oQ(A_n) = 0$ for all $n \in \N$ we still get the countable subadditivity.
\end{remark}

\begin{lemma}\label{lem:outer measure equivalence}
  If $A \subseteq \Omega_{\psi,K}:= \left \{\omega \in \Omega_\psi \,:\,  \|\omega\|_\infty \leq K \right \}$ for $K\in \R_+$, then
  \begin{equation*}
    \overline{Q}(A)\leq \oP(A) \leq (1+ 3 d K +2 d \psi (K))\oQ(A).
  \end{equation*}
\end{lemma}

\begin{proof}
  Since $\mathcal{H}_{\lambda} \subseteq \mathcal{G}_{\lambda+ \epsilon}$ for every $\epsilon>0$, the first inequality holds true. For the second one, assume that $\oQ(A) < \lambda$ for $A \subseteq \Omega_{\psi,K}$. Then there exists a sequence $(G^n)_{n\in \N} \subset \G_\lambda$ such that
  \begin{equation}\label{eq:P-Q-equiv-pr1}
    \liminf_{n\to \infty} \big(\lambda + (G^n\cdot S)_T + \lambda \mathbbm{1}\cdot S_{\rho_\lambda(G^n)}\1_{\{\rho_\lambda(G^n)<\infty\}}\big) \geq \1_A.
  \end{equation}
  Defining the stopping time 
  \begin{equation*}
    \gamma_K(\omega) := \inf \bigl\{ t\in[0,T]\,:\, |S_t(\omega)| \geq K \bigr\}, \quad \omega \in \Omega_\psi,
  \end{equation*}
  we introduce the trading strategy $H^n:= \1_{(0,\rho_\lambda(G^n)\wedge \gamma_K]}\lambda\mathbbm{1} +  \1_{[0,\gamma_K]} G^n$, which satisfies $H^n \in \Hc_{\lambda(1+ 3 d K +2 d \psi (K))}$. Indeed, we observe for $t < \rho_\lambda(G^n)\wedge \gamma_K$ 
  \begin{align*}
    \lambda(1+3 d K +2 d \psi (K)) + (H^n\cdot S)_t 
    &= \lambda(1+3 d K +2 d \psi (K)) + (G^n\cdot S)_t + \lambda \mathbbm{1}\cdot (S_t - S_0)\\
    &\geq \lambda(3 d K +2 d \psi (K))+ \lambda \mathbbm{1}\cdot (S_t - S_0) \\
    &\ge \lambda(3 d K +2 d \psi (K)) - \lambda |\mathbbm{1}\cdot (S_t - S_0)| \\
    &\ge \lambda(3 d K +2 d \psi (K)) - \lambda \sqrt{d} |S_t - S_0|  \ge 0,
  \end{align*}
  and for $t= \rho_\lambda(G^n)\wedge \gamma_K \in (0,T]$ (for $t = 0$ the admissibility is obvious)
  \begin{align*}
    \lambda(1+3 d K +2 d \psi (K))&+ (H^n\cdot S)_{\rho_\lambda(G^n)\wedge \gamma_K} \\
    &=  \lambda(1+3 d K +2 d \psi (K)) + (G^n\cdot S)_{\rho_\lambda(G^n)\wedge \gamma_K} + \lambda\mathbbm{1}\cdot (S_{\rho_\lambda(G^n)\wedge \gamma_K} - S_0)\\
    &\geq \lambda(3 d K +2 d \psi (K)) - \lambda|S_{\rho_\lambda(G^n) \wedge \gamma_K}| + \lambda\mathbbm{1}\cdot (S_{\rho_\lambda(G^n)\wedge \gamma_K} - S_0) \\
    & \geq \lambda(3 d K +2 d \psi (K) - \sqrt{d}K)  + \lambda \sum_{i=1}^d (S^i_{\rho_\lambda(G^n)\wedge \gamma_K} - |S^i_{\rho_\lambda(G^n)\wedge \gamma_K}|) \\
    & = \lambda(3 d K +2 d \psi (K) - \sqrt{d}K)  + \lambda \sum_{i=1}^d 2 S^i_{\rho_\lambda(G^n)\wedge \gamma_K} \1_{S^i_{\rho_\lambda(G^n)\wedge \gamma_K} < 0} \\
    & \ge \lambda(3 d K +2 d \psi (K) - \sqrt{d}K)  - \lambda d 2 (K + \psi(K)) \ge 0,
  \end{align*}
  which extends to $t \in (\rho_\lambda(G^n)\wedge \gamma_K,T]$ because both $\1_{[0,\gamma_K]}G^n$ and $\1_{(0,\rho_\lambda(G^n)\wedge \gamma_K]}$ vanish on that interval. Moreover, from~\eqref{eq:P-Q-equiv-pr1} we get 
  \begin{align*}
    \liminf_{n\to \infty} &\big(\lambda(1+ 3 d K +2 d \psi (K)) +  (H^n \cdot S)_T\big)\\
    &= \liminf_{n\to \infty} \big(\lambda(1+ 3 d K +2 d \psi (K)) + (\1_{[0,\gamma_K]} G^n \cdot S)_T + \lambda \mathbbm{1}\cdot(S_{\rho_\lambda(G^n) \wedge \gamma_K \wedge T} - S_0) \big) 
    \ge \1_A
  \end{align*}
  since $A \subseteq \Omega_{\psi,K}$. Hence, $\oP(A) \leq (1+ 3 d K +2 d \psi (K)) \lambda$, which proves our claim.
\end{proof}

\section{Quadratic variation for typical c\`adl\`ag price paths}\label{sec:quadratic variation}

The existence of the quadratic variation for typical price paths is a crucial ingredient in the construction of the model-free It\^o integral. While the existence was already proven by Vovk in the case of continuous price paths and price paths with jumps with restricted jump size in both directions (see \cite{Vovk2012} resp. \cite{Vovk2015} and cf. Example~\ref{ex:sample spaces}~(1) and (3)), the situation for c\`adl\`ag price paths with an restriction only on jumps directed downwards and thus, in particular, for non-negative c\`adl\`ag price paths was completely unclear so far.\smallskip

In this section we show that the quadratic variation along suitable sequences of partitions exists for typical c\`adl\`ag price paths without any restriction on the jumps directed upwards. We first focus on one-dimensional price paths ($d=1$) and consider
\begin{equation*}
  \Omega_\psi := \left \{ \omega \in \Omega \,:\, \omega(t-)-\omega (t)\leq \psi\bigg(\sup_{s\in[0,t)}|\omega(s)|\bigg)\quad \text{for all } t\in (0,T] \right \}
\end{equation*}
with $\Omega \subseteq D([0,T],\R)$. The extension to general $\Omega_\psi \subseteq D([0,T],\R^d)$ for arbitrary $d\in \N$ is given in Subsection~\ref{subsec:multi-dimensional}. \smallskip

The construction of quadratic variation is based on so-called Lebesgue partitions~$\pi_n(\omega)$. These partitions are generated by stopping times~$(\tau^n_k(\omega))$ acting on the underlying path~$\omega$. Denoting by $\mathbb{D}^n:=\{k2^{-n}\,:\, k\in\mathbb{Z}\}$ the ($n$-th generation of) dyadic numbers for $n\in \N$, $\pi_n(\omega)$~consists of points in time at which the underlying path~$\omega$ crosses  (in space) a dyadic number from~$\mathbb{D}^n$ which is not the same as the dyadic number crossed (as the last number from~$\mathbb{D}^n$) at the  preceding time. This idea is made precise in the next definition.

\begin{definition}\label{def:Lebesgue partition}
  Let $n\in\mathbb{N}$ and let $\mathbb{D}^n:=\{k2^{-n}\,:\, k\in\mathbb{Z}\}$ be the $n$-the generation of dyadic numbers. For a real-valued c\`adl\`ag function $\omega\colon [0,T]\to \R$ its \emph{Lebesgue partition} $\pi_n(\omega):=\{\tau^n_k (\omega)\,:\, k \ge 0\}$ is given by the sequence of stopping times $(\tau^n_k(\omega))_{k\in \mathbb{N}_0}$ inductively defined by
  \begin{equation*}
    \tau^n_{0}(\omega):=0\quad \text{and}\quad D^n_0(\omega):=\max \{\mathbb{D}^n\cap(-\infty,S_0(\omega)]\},
  \end{equation*}
  and for every $k\in \mathbb{N}$ we further set
  \begin{align*}
    \tau^n_k(\omega) &:= \inf \{ t\in[\tau^n_{k-1}(\omega),T]\,:\, \llbracket S_{\tau^n_{k-1}(\omega)}(\omega),S_t(\omega) \rrbracket \cap (\mathbb{D}^n\setminus\{D^n_{k-1}(\omega)\}) \ne \emptyset \},\\
    D^n_k(\omega) &:= \argmin_{ D \in \llbracket S_{\tau^n_{k-1}(\omega)}(\omega),S_{\tau^n_k(\omega)}(\omega) \rrbracket \cap (\mathbb{D}^n\setminus\{D^n_{k-1}(\omega)\} )} | D - S_{\tau^n_k(\omega)}(\omega)|,
  \end{align*}
  with the convention $\inf \emptyset = \infty$ and
  \begin{equation*}
    \llbracket u,v\rrbracket := \begin{cases}
                                  [u,v] & \text{if $u\le v$},\\
                                  [v,u] & \text{if $u>v$}.
                                \end{cases}
  \end{equation*}
\end{definition}

Notice that $D^n_k(\omega)$, $k\in \N$, is $\mathcal{F}_{\tau^n_k}$-measurable function taking values in $\mathbb{D}^n$ and $\tau^n_k(\omega)$ is indeed a stopping time (cf.~\cite[Lemma~3]{Vovk2015}). In the following we often just write~$\tau^n_k$ and~$\pi_n$ instead of~$\tau^n_k(\omega)$ and~$\pi_n(\omega)$, respectively.\smallskip

Along the sequence of Lebesgue partitions we obtain the existence of the quadratic variation for typical c\`adl\`ag price paths. 

\begin{theorem}\label{thm:quadratic variation}
  For typical price paths $\omega\in \Omega_\psi \subseteq D([0,T],\R)$ the discrete quadratic variation 
  \begin{equation*}
    Q^n_t(\omega) := \sum_{k=1}^{\infty} \left( S_{\tau^n_k\wedge t}(\omega) - S_{\tau^n_{k-1}\wedge t}(\omega) \right)^2,\quad t \in [0,T],
  \end{equation*}
  along the Lebesgue partitions $(\pi_n(\omega))_{n\in \N}$ converges in the uniform metric to a function $[S](\omega) \in D([0,T],\R_+)$.
\end{theorem}

\begin{remark}
  Since by \cite[Lemma~3]{Vovk2015} the sequence of partitions $(\pi_n)_{n\in \N}$ is increasing and exhausts the jumps of~$S(\omega)$, \cite[Lemma~2]{Vovk2015} states that the limit $[S](\omega)$ will be a non-decreasing c\`adl\`ag function satisfying $[S]_0(\omega)=0$ and $\Delta[S]_t(\omega) =(\Delta S_t(\omega))^2$ for all $t\in(0,T]$ and all $\omega \in \Omega_\psi$ for which the convergence in Theorem~\ref{thm:quadratic variation} holds. Here we used the notation $\Delta f_t:= f_t-\lim_{s\to t,\,s<t}f_s$ for $f\in D([0,T],\R)$.
\end{remark}

In order to prove Theorem~\ref{thm:quadratic variation}, we first analyze the crossings behavior of typical price paths with respect to the dyadic levels $\mathbb{D}^n$. To be more precise, we introduce the number of upcrossings (resp. downcrossings) of a function $f$ over an open interval $(a,b)\subset \R$. 

\begin{definition}
  Let $f\colon [0,T]\to \R$ be a c\`adl\`ag function, $(a,b)\subset \R$ be an open non-empty interval and $t\in [0,T]$. The number $\U_t^{(a,b)}(f)$ of upcrossings of the interval $(a,b)$ by the function $f$ during the time interval $[0,t]$ is given by 
  \begin{equation*}
    \U_t^{(a,b)}(f):= \sup_{n\in \N} \sup_{0\leq s_1 <t_1 <\dots <s_n<t_n \leq t} \sum_{i=1}^n I(f(s_i),f(t_i)),  
  \end{equation*}
  where 
  \begin{equation*}
    I(f(s_i),f(t_i)):= \begin{cases}
                         1 & \text{if } f(s_i)\leq a \text{ and } f(t_i)\geq b,\\
                         0 & \text{if otherwise}.
                       \end{cases}
  \end{equation*}
  The number $\D_t^{(a,b)}(f)$ of downcrossings is  defined analogously. For $h>0$ we also introduce the accumulated number of upcrossing respectively downcrossings by
  \begin{equation*}
    \U_t(f,h) := \sum_{k\in\mathbb{Z}} \U_t^{(k h,(k+1)h)}(f) \quad \text{and}\quad \D_t(f,h) := \sum_{k\in\mathbb{Z}} \D_t^{(k h,(k+1)h)}(f).
  \end{equation*}
\end{definition}

To derive a deterministic inequality in the spirit of Doob's upcrossing lemma, we use the stopping times
\begin{equation*}
  \gamma_K(\omega) := \inf \left\{ t\in[0,T]\,:\, |S_t(\omega)| \geq K \right\}
\end{equation*}
for $\omega \in \Omega_\psi$ and $K \in \N$. 

\begin{lemma}\label{lem:Doob upcrossing}
  Let $K>0$. For each $n\in\mathbb{N}$, there exists a strongly $1$-admissible simple strategy $H^n \in \mathcal{H}_1$ such that 
  \begin{equation*}
    1+(H^n\cdot S)_{t}(\omega)\ge [2K (2K+\psi(K))]^{-1} 2^{-2n} \U_{t}(\omega,2^{-n})
  \end{equation*}
  for all $t\in [0,T]$ and every $\omega\in \{ \omega \in \Omega_\psi  \,:\, \|\omega\|_\infty < K \}\subseteq D([0,T],\R)$.
\end{lemma}

\begin{proof}
  Let us start by considering the upcrossings $\U_t^{(a,b)}(\omega)$ of an interval $(a,b)\subseteq[-K,K]$. By buying one unit the first time $S_t(\omega)$ drops below ~$a$ and selling the next time $S_t(\omega)$ goes above~$b$ and continuing in this manner until the terminal time $T$ or until we leave the interval $(-K,K)$, whatever occurs first, we obtain a simple strategy $H^{(a,b)} \in \mathcal{H}_{a+K+\psi(K)}$ with
  \begin{equation*}
    a+ K+\psi (K) +  (H^{(a,b)}\cdot S)_{t \wedge \gamma_K}(\omega)
    \ge (b-a)\U_{t \wedge \gamma_K}^{(a,b)}(\omega), \quad (t,\omega) \in [0,T]\times \Omega_\psi.
  \end{equation*}
  For a formal construction of $H^{(a,b)}$ we refer to \cite[Lemma~4.5]{Vovk2015}. Note that we need the predictable bound of the jump size given by~$\psi$ to guarantee the strong admissibility of $H^{(a,b)}$. Set now
  \begin{equation*}
    H^n := [K 2^{n+1}(2K+\psi(K))]^{-1} \sum_{\substack{ k \in \Z,\, (k+1) 2^{-n} < K,\\ k2^{-n}> -K} }  H^{(k2^{-n},(k+1)2^{-n})}.
  \end{equation*}
  Since $H^{(k2^{-n},(k+1)2^{-n})} \in \mathcal{H}_{k 2^{-n}+K+\psi (K)} \subseteq \mathcal{H}_{2K+\psi(K)}$ for all $k$ with $(k+1)2^{-n} < K$ and $k2^{-n}> -K$, we have $H^n \in \mathcal{H}_1$, and
  \begin{align*}
    1+ (H^n \cdot S)_{t (\omega)}(\omega) 
    &\geq [K 2^{n+1}(2K+\psi(K))]^{-1} \sum_{ \substack{ k \in \Z,\, (k+1) 2^{-n} < K,\\ k2^{-n}> -K} } 2^{-n} \U_{t (\omega)}^{(k2^{-n},(k+1)2^{-n})}(\omega)\\
    &= [2 K(2K+\psi(K))]^{-1}2^{-2n} \U_{t}(\omega,2^{-n})
  \end{align*}
for each $t\in [0,T]$ and all $\omega \in \{ \omega \in \Omega_\psi  \,:\, \|\omega\|_\infty < K \}$, which proves the claim.
\end{proof}

With this pathwise version of Doob's upcrossing lemma at hand, we can control the number of level crossings of typical c\`adl\`ag price paths belonging to~$\Omega_\psi$. 

\begin{corollary}\label{cor:crossings}
  For typical price paths $\omega\in\Omega_\psi \subseteq D([0,T],\R)$ there exist an $N(\omega)\in \N$ such that 
  \begin{equation*}
    \U_{T}(\omega,2^{-n})\leq n^2 2^{2n}\quad \text{and}\quad \D_{T}(\omega,2^{-n})\leq n^2 2^{2n}
  \end{equation*}
  for all $n\geq N(\omega)$.
\end{corollary}

\begin{proof}
  Since for each $k \in \mathbb{Z},$ $\U_{t}^{(k2^{-n},(k+1)2^{-n})}(\omega)$ and $\D_{t}^{(k2^{-n},(k+1)2^{-n})}(\omega)$ differ by no more than $1$, we have $|\U_{T}(\omega,2^{-n})-\D_{T}(\omega,2^{-n})| \in [0, 2^{n+1} K]$ for all $n\in \N$ and for every $\omega \in \Omega_\psi$ with $\sup_{t\in [0,T]}|S_t (\omega)|< K$. So if we show that $\oP(B_K) = 0$ for all $K \in \N$, where 
  \begin{equation*}
    B_K := \bigcap_{m \in \N} \bigcup_{n \ge m} A_{K,n}
  \end{equation*}
  with
  \begin{equation*}
    A_{K,n} = \bigg\{ \omega \in \Omega_\psi\,:\, \sup_{t\in [0,T]}|S_t(\omega)|< K \text{ and } \U_{T}(\omega,2^{-n})\geq \frac{n^2 2^{2n}}{2} \bigg\},
  \end{equation*}
  then our claim follows from the countable subadditivity of $\oP$. But using Lemma~\ref{lem:Doob upcrossing} we immediately obtain that $\oP(A_{K,n}) \leq n^{-2} 2 [2 K(2K+\psi(K))]$, and since this is summable, it suffices to apply the Borel-Cantelli lemma (see Lemma~\ref{lem:Borel-Cantelli}) to see $\oP(B_K) = 0$.
\end{proof}

To prove the convergence of the discrete quadratic variation processes $(Q^n)_{n\in \N}$, we shall show that the sequence $(Q^n)_{n\in \N}$ is a Cauchy sequence in the uniform metric on $D([0,T],\R_+)$. For this purpose, we define the auxiliary sequence $(Z^n)_{n\in \N}$ by
\begin{equation*}
  Z^n_t:=Q^n_t-Q^{n-1}_t,\quad t\in [0,T].
\end{equation*}

Similarly as in Vovk~\cite{Vovk2015}, the proof of Theorem~\ref{thm:quadratic variation} is based on the sequence of integral processes $(\mathcal{K}^n)_{n\in \N}$ given by
\begin{equation}\label{eq:integral process}
  \mathcal{K}^n_t :=  n^4 2^{-2n} + 2^{-n+5} (K+\psi(K))^2 + (Z^n_t)^2 -\sum_{k=1}^{\infty} ( Z^n_{\tau^n_k\wedge t} - Z^n_{\tau^n_{k-1}\wedge t} )^2,\quad t\in [0,T],
\end{equation}
for $K\in \N$, and the stopping times
\begin{equation}\label{eq:sigma-n}
  \sigma^n_K := \min \bigg\{ \tau^n_k \,:\, \sum_{i=1}^{k} \big( Z^n_{\tau^n_i} - Z^n_{\tau^n_{i-1}} \big)^2 > n^4 2^{-2n} \bigg\}\wedge\min \left\{ \tau^n_k \,:\,  Z^n_{\tau^n_k} >K \right\},\quad n \in \N.
\end{equation}

The next lemma states that each~$\mathcal{K}^n$ is indeed an integral process with respect to a weakly admissible simple strategy, cf.~\cite[Lemma~5]{Vovk2015}.

\begin{lemma}\label{lem:trading strategy}
  For each $n\in \mathbb{N}$ and $K \in \N$, there exists a weakly admissible simple strategy $L^{K,n} \in \G_{n^4 2^{-2n} + 2^{-n+5} (K+\psi(K))^2}$ such that
  \begin{equation*}
    \mathcal{K}^n_{\gamma_K \wedge\sigma^n_K \wedge t} = n^4 2^{-2n} + 2^{-n+5} (K+\psi(K))^2 + (L^{K,n}\cdot S)_t,\quad t \in [0,T].
  \end{equation*}
\end{lemma}

\begin{proof}
  For each $K\in \N$ and each $n\in\N$ \cite[Lemma~5]{Vovk2015} shows the equality for the strategy
  \begin{equation*}
    L^{K,n}_t := \1_{(0, \gamma_K \wedge\sigma^n_K]}(t) \sum_{k} (-4) Z^n_{\tau^n_k} ( S_{\tau^n_k} - S_{\chi^{n-1}(\tau^n_k)})\1_{(\tau^n_{k},\tau^n_{k+1}]}(t),\quad t \in [0,T],
  \end{equation*}
  where 
  \begin{equation*}
    \chi^{n-1}(t) := \max\left\{\tau^{n-1}_{k'}\;:\;\tau^{n-1}_{k'}\leq t\right\}. 
  \end{equation*}
  Since $L^{K,n}$ is obviously a simple strategy, it remains to prove that 
  \begin{equation}\label{eq:weak admissibility of L}
    L^{K,n} \in \G_{n^4 2^{-2n} + 2^{-n+5} (K+\psi(K))^2}.
  \end{equation}
  First we observe up to time $\tilde\tau^n:=\max\{ \tau_k^n \,:\, \tau_k^n < \gamma_K \wedge \sigma^n_K\}$ that
  \begin{equation}\label{eq:trading strategy pr1}
    \min_{t \in [0,\tilde \tau^n]} \mathcal{K}^n_{\gamma_K \wedge\sigma^n_K\wedge t} \geq  2^{-n+5} (K+\psi(K))^2,
  \end{equation}
  which follows directly from the definition of $\mathcal{K}^n$ and~\eqref{eq:sigma-n}. For $t\in (\tilde\tau^n ,\gamma_K \wedge \sigma^n_K]$ notice that 
  \begin{equation}\label{eq:estimate K}
    |S_{\tilde \tau^n}-S_{\chi^{n-1}(\tilde\tau^n)}|\leq 2^{-n+2},
  \end{equation}
  since we either have $\tilde \tau^n \in \pi_{n-1}$, which implies $\chi^{n-1}(\tilde\tau^n)=\tilde \tau^n $ and $|S_{\tilde \tau^n}-S_{\chi^{n-1}(\tilde\tau^n)}|=0$, or we have $\tilde \tau^n \notin \pi_{n-1}$, which implies \eqref{eq:estimate K} as $\tilde \tau^n < \tau^{n-1}_{k'+1}$ and 
  \begin{equation*}
    |S_{\tilde \tau^n}-S_{\chi^{n-1}(\tilde\tau^n)}|\leq |S_{\tilde \tau^n}-D_{k'}^{n-1}| + |S_{\chi^{n-1}(\tilde\tau^n)} - D_{k'}^{n-1}|\leq 2^{-n+1}+2^{-n+1},
  \end{equation*}
  where $k'$ is such that $\chi^{n-1}(\tilde\tau^n) = \tau^{n-1}_{k'}$. Using~\eqref{eq:estimate K}, $|Z^n_{\tilde\tau^n}| \leq K$ and $|S_{\tilde \tau^n}| \leq K$, we estimate
  \begin{align*}
    |4 Z^n_{\tilde\tau^n} (S_{\tilde\tau^n} - S_{\chi^{n-1}(\tilde\tau^n)}) ( S_t - S_{\tilde\tau^n} )| 
    & \leq 4 K 2^{-n+2} (|S_t| + K) = 2^{-n+4} (K |S_t| + K^2), 
  \end{align*}
  which together with~\eqref{eq:trading strategy pr1} gives weak admissibility as claimed in~\eqref{eq:weak admissibility of L}.
\end{proof}

\begin{corollary}\label{cor:integral process}
  For typical price paths $\omega\in\Omega_\psi\subseteq D([0,T],\R)$ there exist an $N(\omega)\in \N$ such that 
  \begin{equation*}
    \mathcal{K}^n_{\gamma_K \wedge\sigma^n_K \wedge t} (\omega) < n^6 2^{-n},\quad t\in [0,T],
  \end{equation*}
  for all $n\geq N(\omega)$.
\end{corollary}

\begin{proof}
  Consider the events
  \begin{equation*}
    A_{n,m} := \bigg\{\omega \in \Omega_\psi \,:\, \exists t \in [0,T] \text{ s.t. } 
    \mathcal{K}^n_{\gamma_K \wedge\sigma^n_K \wedge t} (\omega) \ge n^6 2^{-n} \text{ and } \sup_{t\in [0,T]}|S_t(\omega)|\leq m \bigg\}
  \end{equation*}
  for $n,m \in \N$. By the countable subadditivity of $\oP$ and the Borel-Cantelli lemma (see Lemma~\ref{lem:Borel-Cantelli}) the claim follows once we have shown that $\sum_n \oP(A_{n,m}) < \infty$ for every $m\in \N$. To that end, we define the stopping times
  \begin{equation*}
    \rho^n := \inf \left\{ t\in [0,T]\,:\,  \mathcal{K}^n_{\gamma_K \wedge\sigma^n_K \wedge t}\geq  n^6 2^{-n}\right\}, \quad n\in \N,
  \end{equation*}
  so that
  \begin{equation*}
    A_{n,m} = \Big\{\omega \in \Omega_\psi \,:\,  n^{-6} 2^n \mathcal{K}^n_{\gamma_K \wedge\sigma^n_K \wedge \rho^n \wedge T} (\omega) \geq 1 \text{ and } \sup_{t\in [0,T]}|S_t(\omega)|\leq m  \Big\}.
  \end{equation*}
  Now it follows directly from Lemma~\ref{lem:trading strategy} that
  \begin{equation*}
    \oQ(A_{n,m}) 
    \leq n^{-6} 2^n (n^4 2^{-2n} + 2^{-n+5} (K+\psi(K))^2) 
    = n^{-2} 2^{-n} + n^{-6}2^5 (K+\psi(K))^2,
  \end{equation*}
  which is summable. Since $\oP(A_{n,m}) \leq (1+ 3m+2\psi (m)) \oQ(A_{n,m})$ by Lemma~\ref{lem:outer measure equivalence}, the proof is complete.
\end{proof}

Finally, we have collected all necessary ingredients to prove the main result of this section, namely Theorem~\ref{thm:quadratic variation}. More precisely, we shall show that $(Q^n-Q^{n-1})_{n\in \N}$ is a Cauchy sequence. This implies Theorem~\ref{thm:quadratic variation} since the uniform metric on $D([0,T],\R_+)$ is complete.

\begin{proof}[Proof of Theorem~\ref{thm:quadratic variation}]
  For $K\in \N$ let us define  
  \begin{equation*}
    A_K :=\bigg\{ \omega \in \Omega_\psi \,:\, \sup_{t\in [0,T]}|S_t(\omega)|\leq K \text{ and } \sup_{t\in [0,T]}|Z^n_t(\omega)| \geq n^3 2^{-\frac{n}{2}} \text{ for infinitely many } n\in \N\bigg\}
  \end{equation*}
  and 
  \begin{align*}
    B :=\bigg\{ \omega \in \Omega_\psi \,:\,& \exists N(\omega)\in \N \text{ s.t. }\mathcal{K}^n_{\gamma_K \wedge\sigma^n_K \wedge t} (\omega) < n^6 2^{-n},\quad t\in [0,T],\\ 
    &\U_{T}(\omega,2^{-n})\leq n^2 2^{2n} \text{ and } \D_{T}(\omega,2^{-n})\leq n^2 2^{2n},\quad n\geq N(\omega)\bigg\}.
  \end{align*}
  Thanks to the countable subadditivity of $\overline{P}$ it is sufficient to show that $\overline{P}(A_K)=0$ for every $K\in \N$. Moreover, again by the subadditivity of $\overline{P}$ we see  
  \begin{equation*}
    \overline{P}(A_K)\leq  \overline{P}(A_K\cap B) + \overline{P}(A_K\cap B^c).
  \end{equation*}
  By Corollary~\ref{cor:crossings} and Corollary~\ref{cor:integral process} it is already known that $\overline{P}(A_K\cap B^c)=0$. In the following we show that $A_K\cap B=\emptyset$. 
  
  For this purpose, let us fix an $\omega \in B$ such that $\sup_{t\in [0,T]}|S_t(\omega)|\leq K$. Since $\omega \in B$ there exits an $N(\omega)\in \N$ such that for  all $m\ge N(\omega)$:
  \begin{enumerate}
    \item[(a)] The number of stopping times in $\pi_m$ does not exceed $2m^2 2^{2m}+2\leq 3m^2 2^{2m}$.
    \item[(b)] The number of stopping times in $\pi_m$ such that 
         \begin{equation*}
           \big|\Delta  S_{\tau^m_k}(\omega)\big|:=\bigg|S_{\tau^m_k}(\omega)-\lim_{s\to\tau^m_k,\,s<\tau^m_k} S_s(\omega)\bigg|\geq 2^{-m+1}, \quad\tau^m_k \in\pi_m, 
         \end{equation*}
         is less or equal to $2m^2 2^{2m}$.
  \end{enumerate}
  As $\sup_{t\in [0,T]}|S_t(\omega)|\leq K$, notice that $\gamma_K(\omega)=T$ and that for $t\in [0,T]$ we have
  \begin{align*}
    Z^n_{\tau^n_{k+1}\wedge t}(\omega) - &Z^n_{\tau^n_{k}\wedge t}(\omega)
    = \big(Q^n_{\tau^n_{k+1}\wedge t}(\omega) - Q^n_{\tau^n_{k}\wedge t}(\omega)\big)-\big(Q^{n-1}_{\tau^n_{k+1}\wedge t}(\omega) - Q^{n-1}_{\tau^n_{k}\wedge t}(\omega)\big)\\
    =& \big(S_{\tau^n_{k+1}\wedge t}(\omega) - S_{\tau^n_{k}\wedge t}(\omega)\big)^2\\
     &-\big(\big(S_{\tau^n_{k+1}\wedge t}(\omega) - S_{\chi^{n-1}(\tau^n_{k}\wedge t)}(\omega)\big)^2-\big(S_{\tau^n_{k}\wedge t}(\omega) - S_{\chi^{n-1}(\tau^n_{k}\wedge t)}(\omega)\big)^2\big)\\
    =& -2 \big(S_{\tau^n_{k}\wedge t}(\omega) - S_{\chi^{n-1}(\tau^n_{k}\wedge t)}(\omega)\big) \big(S_{\tau^n_{k+1}\wedge t}(\omega) - S_{\tau^n_{k}\wedge t}(\omega)\big),
  \end{align*}
  where we recall that $\chi^{n-1}(t) := \max\{\tau^{n-1}_{k'}\;:\;\tau^{n-1}_{k'}\le t\}$. Therefore, keeping \eqref{eq:estimate K} in mind, the infinite sum in \eqref{eq:integral process} can be estimated by
  \begin{align}\label{eq:infinite sum}
    \sum_{k=0}^{\infty} \bigg(Z^n_{\tau^n_{k+1}\wedge t}(\omega) - Z^n_{\tau^n_{k}\wedge t}(\omega) \bigg)^2  
    &= 4 \sum_{k=0}^{\infty} \big(S_{\tau^n_{k}\wedge t}(\omega) - S_{\chi^{n-1}(\tau^n_{k}\wedge t)}(\omega)\big)^2 \big(S_{\tau^n_{k+1}\wedge t}(\omega) - S_{\tau^n_{k}\wedge t}(\omega)\big)^2 \notag\\
    &\leq 2^{6-2n} \sum_{k=0}^{\infty} \big(S_{\tau^n_{k+1}\wedge t}(\omega) - S_{\tau^n_{k}\wedge t}(\omega)\big)^2.
  \end{align}
   
  For $n\ge N = N(\omega)$ and $t\in [0,T]$ we observe for the summands in \eqref{eq:infinite sum} the following bounds, which are similar to the bounds (A)-(E) in the proof of~\cite[Theorem~1]{Vovk2015}:
  \begin{enumerate}
    \item If $\tau^n_{k+1}\notin\pi_{n-1}$, then one has $\chi^{n-1}(\tau^n_{k+1}) = \chi^{n-1}(\tau^n_{k}) = \tau^{n-1}_{k'}$ for some $k'$ and thus
          \begin{equation*}
            \big| S_{\tau^n_{k+1}\wedge t}(\omega) - S_{\tau^n_{k}\wedge t}(\omega)\big| \leq |S_{\tau^n_{k+1}\wedge t}(\omega) -D_{k'}^{n-1}|+|S_{\tau^n_{k}\wedge t}(\omega) -D_{k'}^{n-1}| \leq 2^{2-n}.
	  \end{equation*}
          The number of such summands is at most $3n^2 2^{2n}$.
    \item If $\tau^n_{k+1}\in\pi_{n-1}$ and $|\Delta S_{\tau^n_{k+1}}| \leq 2^{-n+1}$, then one has
          \begin{equation*}
            \big| S_{\tau^n_{k+1}\wedge t}(\omega) - S_{\tau^n_{k}\wedge t}(\omega)\big| \leq 2^{1-n} + 2^{-n+1} = 2^{2-n}
          \end{equation*}
          and the number of such summands is at most $3n^2 2^{2n}$.
    \item If $\tau^n_{k+1}\in\pi_{n-1}$ and $|\Delta S_{\tau^n_{k+1}}| \in [2^{-m+1},2^{-m+2})$, for some $m\in\{N,N+1,\dots,n\}$ than one has that
          \begin{equation*}
            \big| S_{\tau^n_{k+1}\wedge t}(\omega) - S_{\tau^n_{k}\wedge t}(\omega)\big| \leq 2^{1-n} + 2^{-m+2}.
          \end{equation*}
          and the number of such summands is at most $2m^2 2^{2m}$.
    \item If $\tau^n_{k+1}\in\pi_{n-1}$ and $\Delta S_{\tau^n_{k+1}} \geq 2^{-N+2}$, then one has
          \begin{equation*}
            \big| S_{\tau^n_{k+1}\wedge t}(\omega) - S_{\tau^n_{k}(\omega)\wedge t}\big| \leq 2K
	  \end{equation*}
          and the number of such summand is bounded by a constant $C=C(\omega,K)$ independent of $n$.
  \end{enumerate}
  Using the bounds derived in (1)-(4), the estimate~\eqref{eq:infinite sum} can be continued by
  \begin{align*}
    \sum_{k=0}^{\infty} \bigg( Z^n_{\tau^n_{k+1}\wedge t}(\omega)& - Z^n_{\tau^n_{k}\wedge t}(\omega) \bigg)^2\\
    &\leq 2^{6-2n} \bigg( 6 n^2 2^{2n} 2^{4-2n}+ \sum_{m=N}^n 2m^2 2^{2m} (2^{1-n} + 2^{-m+2})^2+4 C K^2 \bigg),
  \end{align*}
  and thus there exists an $\tilde N = \tilde N(\omega) \in \N$ such that
  \begin{equation*}
    \sum_{k=0}^{\infty} \bigg( Z^n_{\tau^n_{k+1}\wedge t}(\omega) - Z^n_{\tau^n_{k}\wedge t}(\omega) \bigg)^2 \leq  2^{-2n} n^4,\quad t \in [0,T],
  \end{equation*}
  for all $n\geq \tilde N$. Combining the last estimate with the definition of $\mathcal{K}^n$ (cf.~\eqref{eq:integral process}), we obtain
  \begin{equation*}
    \mathcal{K}^n_{\sigma^n_K \wedge t}(\omega) \geq \big(Z^n_{\sigma^n_K\wedge t}(\omega)\big)^2, \quad t\in [0,T],
  \end{equation*}
  for all $n\geq \tilde N$. Moreover, by assumption on $\omega$ one has 
  \begin{equation*}
    \mathcal{K}^n_{\sigma^n_K \wedge t}(\omega) < n^6 2^{-n}, \quad t\in [0,T],
  \end{equation*}
  for all $n\geq N \vee \tilde N$. In particular, we conclude that $\sup_{t\in [0,T]}\big|Z^n_{\sigma^n_K\wedge \gamma_K\wedge t}(\omega)\big|<K$ whenever $n$ is large enough and thus
  \begin{equation*}
    n^6 2^{-n}> \big(Z^n_{t}(\omega)\big)^2, \quad t\in [0,T],
  \end{equation*}  
  for all sufficiently large $n$. Finally, we have $\sup_{t\in [0,T]}|Z^n_t(\omega)| < n^3 2^{-\frac{n}{2}}$ for all large $n$ and therefore $\omega\notin A_K\cap B$.
\end{proof}

\begin{remark}
  The existence of quadratic variation in the sense of Theorem~\ref{thm:quadratic variation} is equivalent to the existence of quadratic variation in the sense of F\"ollmer (see \cite[Proposition~3]{Vovk2015}). Therefore, Theorem~\ref{thm:quadratic variation} opens the door to apply F\"ollmer's pathwise It\^o formula~\cite{Follmer1979} to typical price paths belonging to the sample space~$\Omega_\psi$ and in particular to define the pathwise integral $\int f^\prime (S_s)\d S_s$ for $f\in C^{2}$ or for more general path-dependent functionals as shown by Cont and Fourni\'e~\cite{Cont2010a}, Imkeller and Pr\"omel~\cite{Imkeller2015}, and Ananova and Cont~\cite{ Ananova2017}.
\end{remark}

\subsection{Extension to multi-dimensional price paths}\label{subsec:multi-dimensional} 

In order to extend the existence of quadratic variation from one-dimensional to multi-dimensional typical price paths, we consider now the sample space $\Omega_\psi\subseteq D([0,T],\R^d)$ and introduce a $d$-dimensional version of the Lebesgue partitions for~$d\in \N$.

\begin{definition}\label{def:multi Lebesgue partition}
  For $n\in \N$ and a $d$-dimensional c\`adl\`ag function $\omega\colon [0,T]\to \R^d$ its \emph{Lebesgue partition} $\pi_n(\omega):=\{\tau^n_k (\omega)\,:\, k \ge 0\}$ is iteratively defined by $\tau^n_0(\omega):=0$  and
  \begin{equation*}
    \tau^n_k(\omega) := \min\bigg \{ \tau > \tau^n_{k-1}(\omega)\,:\, \tau \in \bigcup_{i=1}^d \pi_n(\omega^i)\cup  \bigcup_{i,j=1,i\neq j}^d \pi_n(\omega^i+\omega^j) \bigg\},\quad k\in \N,
  \end{equation*}
  where $\omega = (\omega^1,\dots,\omega^d)$ and $\pi_n(\omega^i)$ and $\pi_n(\omega^i +\omega^j)$ are the Lebesgue partitions of $\omega^i$ and $\omega^i+\omega^j$ as introduced in Definition~\ref{def:Lebesgue partition}, respectively.
\end{definition}

To state the existence of quadratic variation for typical price paths in $\Omega_\psi$, we define the canonical projection on $\Omega_\psi$ by $S^i_t(\omega):=\omega^i(t)$ for $\omega = (\omega^1,\dots, \omega^d)\in \Omega_\psi$, $t\in [0,T]$ and $i=1,\dots,d$.

\begin{corollary}\label{cor:quadratic variation}
  Let $d\in \N $ and $1\leq i,j\leq d$. For typical price paths $\omega\in \Omega_\psi$ the discrete quadratic variation 
  \begin{equation*}
    Q^{i,j,n}_t(\omega) := \sum_{k=1}^{\infty} \left( S^i_{\tau^n_k\wedge t}(\omega) - S^i_{\tau^n_{k-1}\wedge t}(\omega) \right) \left( S^j_{\tau^n_k\wedge t}(\omega) - S^j_{\tau^n_{k-1}\wedge t}(\omega) \right),\quad t \in [0,T],
  \end{equation*}
  converges along the Lebesgue partitions $(\pi_n(\omega))_{n\in \N}$ in the uniform metric to a function $[S^i,S^j](\omega) \in D([0,T],\R)$.
\end{corollary}

\begin{proof}
  To show the convergence of $Q^{i,j,n}_\cdot(\omega)$ for a path $\omega \in \Omega_\psi$, we observe that
  \begin{align*}
    S^i_{s,t}(\omega) S^j_{s,t}(\omega)=\frac{1}{2} \left(\left((S^i_t(\omega) + S^j_t(\omega))-(S^i_s(\omega)+S^j_s(\omega))\right)^2 -(S^i_{s,t}(\omega))^2-(S^j_{s,t}(\omega))^2  \right)
  \end{align*}
  for $s,t\in[0,T]$ and thus it is sufficient to prove the existence of the quadratic variation of $S^i(\omega)$ and $S^i(\omega) + S^j(\omega)$ for $1\leq i,j\leq d$ with $i\neq j$. For typical price paths this can be done precisely as in the proof of Theorem~\ref{thm:quadratic variation} with the only exception that the bounds (a)-(b) and (1)-(4) change by a multiplicative constant depending only on the dimension $d$.
\end{proof}

\section{Model-free It\^o integration}\label{sec:ito integration}

The key problem of ``stochastic'' integration with respect to typical price paths is, unsurprisingly, that they are not of bounded variation and that there exists no reference probability measure in the present probability-free setting.

The model-free It\^o integral provided in this section is a pathwise construction and comes with two natural interpretations in financial mathematics: in the case of existence the integral has an interpretation as the capital process of an adapted trading strategy and the set of price paths where the integral does not exists allows for a model-free arbitrage opportunity of the first kind (cf. Proposition~\ref{prop:na1 interpretation}).

Roughly speaking, the construction of the model-free It\^o integral is based on the existence of the quadratic variation for typical price paths and on an application of the pathwise Hoeffding inequality due to Vovk~\cite{Vovk2012} in the case of continuous price trajectories, and of the pathwise Burkholder-Davis-Gundy inequality  due to Beiglb\"ock and Siorpaes~\cite{Beiglbock2015} in the case of price paths with jumps. As in the classical approach to stochastic integration, we first define the It\^o integral for simple integrands and extend it via an approximation scheme to a larger class of integrands. Our simple integrands are the step functions:\smallskip

A process $F\colon \Omega_\psi \times [0,T] \to \R^d$ is called \emph{step function} if $F$ is given by  
\begin{equation}\label{eq:step function}
  F_t :=F_0 \1_{\{0\}}(t)+ \sum_{i=0}^{\infty} F_{\sigma_i} \1_{(\sigma_i,\sigma_{i+1}]}(t),\quad t\in [0,T],
\end{equation}
where $(\sigma_i)_{i\in\N_0}$ is an increasing sequence of stopping times such that for each $\omega \in \Omega_\psi$ there exists an $N(\omega)\in\N$ with $\sigma_i(\omega)=\sigma_{i+1}(\omega)$ for all $i\geq N(\omega)$, $F_0\in \R^d$ and $F_{\sigma_i}\colon \Omega_\psi \to \R^d$ is $\mathcal{F}_{\sigma_i}$-measurable. For such a step function $F$ the corresponding integral process~$(F \cdot S)_t$ is well-defined for all~$(t,\omega) \in [0,T]\times \Omega_\psi$ and we recall that
\begin{equation}\label{eq:step function integral}
  (F \cdot S)_t := \sum_{i=0}^{\infty} F_{\sigma_i} S_{\sigma_i\wedge t,\sigma_{i+1}\wedge t},\quad t\in [0,T],
\end{equation}
and $ S_{\sigma_i\wedge t,\sigma_{i+1}\wedge t} :=  S_{\sigma_{i+1}\wedge t}- S_{\sigma_{i+1}\wedge t}$. Throughout the whole section we denote by $(\pi_n(\omega))_{n\in \N}$ the sequence of Lebesgue partitions consisting of the stopping times $(\tau^n_k(\omega))_{k\in \N}$ as introduced in Definition~\ref{def:multi Lebesgue partition} for $\omega\in \Omega_\psi$, and the quadratic variation matrix of $\omega$ along $(\pi_n(\omega))_{n\in \N}$  is given by
\begin{equation*}
  [S]_t (\omega):= \big([S^i,S^j]_t(\omega)\big)_{1\leq i,j\leq d},\quad t\in [0,T],
\end{equation*}
where $S^i_t(\omega):= \omega^i(t)$ for $\omega=(\omega^1,\dots,\omega^d)$ and we refer to Corollary~\ref{cor:quadratic variation} for the definition of $[S^i,S^j]_t(\omega)$. Recall that if the quadratic variation exists as a uniform limit, then it exists also in the sense of F\"ollmer along the same sequence of partitions (see \cite[Proposition~3]{Vovk2015}). Hence, one gets 
\begin{equation*}
  \int_0^t F^{\otimes 2}_s \d [S]_s := \sum_{i,j=1}^d \int_0^t F^i_s F^j_s \d [S^i, S^j]_s 
  := \liminf_{n\to \infty }\sum_{k=0}^{\infty} \sum_{i,j=1}^d F_{\tau_k^n}^iF_{\tau_k^n}^j S^i_{\tau_k^n\wedge t,\tau_{k+1}^n\wedge t}S^j_{\tau_k^n\wedge t,\tau_{k+1}^n\wedge t}, 
\end{equation*}
for $t \in [0,T]$, is actually a true limit for typical price paths. 

\begin{remark}
  The existence of the quadratic variation along the Lebesgue partitions is ensured for typical price paths belonging to the path space~$\Omega_\psi$ by Corollary~\ref{cor:quadratic variation}. For some special cases of~$\Omega_\psi$, such as the space $\Omega_c$ of continuous paths or the space $\tilde \Omega_\psi$ of c\`adl\`ag paths with mildly restricted jumps, the construction of the quadratic variation can be also found in \cite[Lemma~8.1]{Vovk2012} resp. \cite[Theorem~2]{Vovk2015}, cf. Example~\ref{ex:sample spaces}.
\end{remark}

In the following we identify two processes $X,Y\colon \Omega_\psi \times [0,T]\to \R^d$ if for typical price paths we have $X_t = Y_t$ for all $t\in [0,T]$, and we write $\overline L_0(\R^d)$ for the resulting space of equivalence classes which is equipped with the distance
\begin{equation*}
  d_\infty(X,Y) := \overline E[\|X-Y\|_\infty \wedge 1],
\end{equation*}
where $\lVert f \rVert_\infty := \sup_{t\in [0,T]}|f(t)|$ for $f \colon [0,T] \to \R^d$ denotes the supremum norm and $\overline E$ denotes an expectation operator defined for $Z\colon\Omega_\psi \rightarrow [0, \infty]$ by 
\begin{equation*}
  \overline E[Z] := \inf \left\{\lambda > 0\,: \,\exists (H^n)_{n\in \N} \subseteq \mathcal{H}_{\lambda} \text{ s.t. } \forall \omega \in \Omega_\psi\,: \, \liminf_{n \to\infty} (\lambda + (H^n\cdot S)_T(\omega)) \ge Z(\omega) \right\}.
\end{equation*}
As in~\cite[Lemma~2.11]{Perkowski2016} it can be shown that $(\overline L_0(\R^d), d_\infty)$ is a complete metric space and $(\overline{\mathcal{D}}(\R^d), d_\infty)$ is a closed subspace, where $\overline{\mathcal{D}}(\R^d)$ are those processes in $\overline L_0(\R^d)$ which have a c\`adl\`ag representative.\smallskip

The main results about model-free It\^o integration are summarized in the following theorem.

\begin{theorem}\label{thm:integral}
  There exists two metric spaces $(\overline{H}_1, d_{\overline{H}_1})$ and $(\overline{H}_2, d_{\overline{H}_2})$ such that the (equivalence classes of) step functions are dense in $\overline{H}_1$, $\overline{H}_2$ embedds into $\overline{\mathcal{D}}(\R^d)$ and the integral map $I \colon F \mapsto (F \cdot S)$, defined for step functions in~\eqref{eq:step function integral}, has a continuous extension that maps $(\overline{H}_1, d_{\overline{H}_1})$ to $(\overline{H}_2, d_{\overline{H}_2})$. Moreover, one has the following continuity estimates:
  \begin{enumerate}
    \item[(i)] For $\Omega_\psi= \Omega_c$ the integral map $I$ satisfies~\eqref{eq:continuity continuous 1} and~\eqref{eq:continuity continuous 2} thus one can define $d_{\overline{H}_1}=d_{\mathrm{QV}}$ (which is defined by formula \eqref{eq:distance continuous 1}) and $d_{\overline{H}_2}=d_{\infty}$  or $d_{\overline{H}_1}=d_{\mathrm{QV}, loc}$ and  $d_{\overline{H}_2}=d_{\infty, loc}$ (which are defined in \eqref{eq:distance continuous 2}).    
    \item[(ii)] For general $\Omega_\psi$ the integral map $I$ satisfies~\eqref{eq:continuity cadlag} and one can define $d_{\overline{H}_1}=d_{\infty}$ and $d_{\overline{H}_2}=d_{\infty,\psi}$ (defined in~\eqref{eq:distance cadlag}). 
  \end{enumerate}
\end{theorem}

Let us briefly comment on the spaces of integrands covered by the model-free It\^o integral of Theorem~\ref{thm:integral}. \medskip

\begin{remark}~
  \begin{enumerate}
    \item It is easy to verify that the metric space~$(\overline{H}_1, d_{\overline{H}_1})$ can be chosen to contain the left-continuous versions of adapted c\`adl\`ag processes, cf. \cite[Theorem~3]{Karandikar1995} and \cite[Theorem~3.5]{Perkowski2016}.
    \item If we replace the filtration $\mathcal{F}_t$ by its right-continuous version, we can define $(\overline{H}_1, d_{\overline{H}_1})$ such that it contains at least the c\`agl\`ad adapted processes and, furthermore, such that if $(F_n) \subset \overline{H}_1$ is a sequence with $\sup_{\omega \in \Omega_\psi} \| F_n(\omega) - F(\omega) \|_\infty \to 0$, then $F\in \overline{H}_1$ and there exists a subsequence $(F_{n_k})$ with 
    \begin{equation*}
      \lim_{k \to \infty} \| (F_{n_k} \cdot S)(\omega) - (F \cdot S)(\omega) \|_\infty = 0
    \end{equation*}
    for typical price paths $\omega\in \Omega_\psi$. 
  \end{enumerate}
  In both cases we can take $(\overline{H}_1, d_{\overline{H}_1})$ as the closure of step functions with respect to~$d_{\infty}$. 
\end{remark}

\begin{remark} \label{q_var_part_ind}
  In a recent work~\cite{Vovk2016} Vovk introduces a related (but less systematic) approach to define model-free It\^o integrals. He obtains the convergence of non-anticipating Riemann sums along a suitably chosen sequence of partitions and the limit is interpreted as a model-free It\^o integral. However, for this construction no continuity estimates are given and (a priori) the limit might depend on the chosen sequence of partitions. Furthermore, Vovk works with different techniques and on a more restrictive sample space, compared to the present work. 
  
  Thanks to the continuity estimates, our integral $(F\cdot S)$ is independent of the approximating sequence $(F^n \cdot S)$ of step functions $F^n$. However, to define the quadratic variation we work with the sequence of Lebesgue partitions. But it follows from Theorem~\ref{thm:integral} together with It\^o's formula that the quadratic variation along any sequence of partitions of stopping times $(\pi_n = \{\tau^n_k: k \in \N_0\})_{n \in \N}$ for which $S^n = \sum_{k=0}^\infty S_{\tau^n_k} \1_{(\tau^n_k, \tau^n_{k+1}]}$ converges to $S$ in $d_{\infty}$ agrees with $[S]$ for typical price paths.
\end{remark}

To prove Theorem~\ref{thm:integral}, we will derive in the following subsections suitable continuity estimates for the integrals of step functions. Given these continuity estimates, the proof of Theorem~\ref{thm:integral} follows directly by approximating general integrands by step functions.

\subsection{Integration for continuous paths}\label{subsec:integration continuous}

In this subsection we focus on the sample space $\Omega_c :=C([0,T],\R^d)$ consisting of continuous paths $\omega\colon [0,T]\to \R^d$.

We recover essentially the results of~\cite[Theorem~3.5]{Perkowski2016} and are able to construct our integral for c\`agl\`ad adapted integrands. However, in~\cite{Perkowski2016} we worked with the uniform topology on the space of integrands while here we are able to strengthen our results and to replace the uniform distance with a rather natural distance that depends only on the integral of the squared integrand against the quadratic variation. We are able to show that the closure of the step functions in this new distance contains the c\`agl\`ad adapted processes. However, in principle this closure might contain a wider class of integrands.\smallskip

The main ingredient in our construction is the following continuity estimate for the pathwise stochastic integral of a step function. It is based on Vovk's pathwise Hoeffding inequality.

\begin{lemma}[Model-free concentration of measure, continuous version]\label{lem:model free ito cont}
  Let $F\colon\Omega_c\times [0,T]\to \R^d $ be a step function. Then we have for all $a,b> 0$
  \begin{equation*}
    \oP\left(\big\{\lVert (F\cdot S)\rVert_\infty \ge a \sqrt{b} \big\} \cap \Big\{ \int_0^T F_s^{\otimes 2} \d [ S]_s \le b \Big\}\right) \le 2\exp(-a^2/2).
  \end{equation*}
\end{lemma}

\begin{proof}
  Let $F_t = F_0 \1_{0}(t) + \sum_{m=0}^\infty F_m \1_{(\sigma_m,\sigma_{m+1}]}(t)$. For $n \in \N$ we define the stopping times
  \begin{equation*}
    \zeta^n_0:=0, \qquad \zeta^n_{k+1} := \inf\{ t \ge \zeta^n_k: | (F \cdot S)_{\zeta^n_k, t}| = 2^{-n} \},
  \end{equation*}
  and also $\tau^n_0:=0$, $\tau^n_{k+1} := \inf\{ t \ge \tau^n_k : |S_{\tau^n_k,t}| = 2^{-n}\}$. We then write
  \begin{equation*}
    \rho^n_0 := 0, \qquad \rho^n_{k+1} := \inf\{ t > \rho^n_k: t = \zeta^{2n}_i \text{ or } t = \tau^n_i \text{ or } t = \sigma_i \text{ for some } i \ge 0 \}
  \end{equation*}
  for the union of the $(\zeta^{2n}_k)_k$ and $(\tau^n_k)_k$ and $(\sigma_m)_m$.  By definition of the times $(\rho^n_k)$ we have
  \begin{equation*}
    \sup_{t\in [0,T]}\big| (F \cdot S)_{\rho^n_k\wedge t, \rho^n_{k+1} \wedge t} \big| \le 2^{-2n}
  \end{equation*}
  and $F$ is constant on $(\rho^n_k,\rho^n_{k+1}]$ for all $k$, and therefore Vovk's pathwise Hoeffding inequality, \cite[Theorem~A.1]{Vovk2012} or \cite[Lemma~A.1]{Perkowski2016}, gives us for every $\lambda \in \R$ a strongly 1-admissible simple strategy $H^{\lambda,n} \in \Hc_{1}$ such that
  \begin{equation}\label{eq:model free ito cont pr1}
    1 + (H^{\lambda,n} \cdot S)_t \ge \exp\bigg( \lambda (F\cdot S)_{t} - \frac{\lambda^2}{2} \sum_{k=0}^\infty 2^{-4n} \1_{\{\rho^n_k \le t\}} \bigg) =: \mathcal{E}^{\lambda,n}_{t}, \qquad t \in [0,T].
  \end{equation}
  Next, observe that for all $i =1 , \dots, d$
  \begin{equation*}
    \sup_{t \in [0,T]} \Big| \sum_{k=0}^\infty S^i_{\rho^n_k} \1_{[\rho^n_k,\rho^n_{k+1})}(t) - S^i_t \Big| \le 2^{-n},
  \end{equation*}
  so since $2^{-n}$ decays faster than logarithmically,~\cite[Corollary~3.6]{Perkowski2016} shows that for typical price paths we have for $i,j =1, \dots, d$
  \begin{equation*}
    \lim_{n \to \infty} \sup_{t \in [0,T]} \Big|\sum_{k=0}^\infty S^i_{\rho^n_k \wedge t, \rho^n_{k+1}\wedge t} S^j_{\rho^n_k \wedge t, \rho^n_{k+1}\wedge t} - [S^i, S^j]_t \Big| = 0,
  \end{equation*}
  and using that $F(\omega)$ is piecewise constant for all $\omega \in \Omega_c$, we also get
  \begin{equation}\label{eq:model free ito cont pr2}
    \lim_{n \to \infty} \sup_{t \in [0,T]} \Big|\sum_{k=0}^\infty F^i_{\rho^n_{k+}} F^j_{\rho^n_{k+}} S^i_{\rho^n_k \wedge t, \rho^n_{k+1}\wedge t} S^j_{\rho^n_k \wedge t, \rho^n_{k+1}\wedge t} - \int_0^t F^i_s F^j_s \d [S^i, S^j]_s \Big| = 0
  \end{equation}
  for typical price paths, where $F^i_{\rho^n_{k+}}$ is simply the value that $F^i$ attains on $(\rho^n_k, \rho^n_{k+1}]$. We proceed by estimating for $k \ge 0$
  \begin{equation*}
    2^{-2n} \le \big| (F \cdot S)_{\rho^n_k, \rho^n_{k+1}} \big| + 2^{-2n} \1_{\{\rho^n_k \text{ or } \rho^n_{k+1} = \tau^n_i \text{ or } \sigma_i \text{ for some }i \ge 0\}},
  \end{equation*}
  which together with~\eqref{eq:model free ito cont pr2} leads to 
  \begin{align*}
    \limsup_{n \to \infty} \sum_{k=0}^\infty 2^{-4n} \1_{\{\rho^n_k \le t\}} & \le \limsup_{n \to \infty} \Big(\sum_{k=0}^\infty \big| (F \cdot S)_{\rho^n_k \wedge t, \rho^n_{k+1} \wedge t} \big|^2 \\
    &\hspace{75pt} + 3 \times 2^{-4n} \times (|\{k: \sigma_k \le t \}| + |\{k: \tau^n_k \le t\}|)\Big) \\
    & = \limsup_{n \to \infty} \Big(\sum_{k=0}^\infty \sum_{i,j=1}^d F^i_{\rho^n_{k+}} F^j_{\rho^n_{k+}} S^i_{\rho^n_k \wedge t, \rho^n_{k+1}\wedge t} S^j_{\rho^n_k \wedge t, \rho^n_{k+1}\wedge t} \\
    &\hspace{100pt} + 3 \times 2^{-4n} \times |\{k: \tau^n_k \le t\}| \Big) \\
    & = \int_0^t F^{\otimes 2}_s \d [ S]_s + \limsup_{n \to \infty} 3 \times 2^{-4n} \times |\{k: \tau^n_k \le t\}|
  \end{align*}
  for typical price paths. For typical price paths we also have 
  \begin{equation*}
    \lim_{n \to \infty} 2^{-2n} \times |\{k: \tau^n_k \leq t\}| = \sum_{i=1}^d [S^i,S^i]_t
  \end{equation*}
  and consequently
  \begin{equation}\label{eq:model free ito cont pr3}
    \limsup_{n \to \infty} \sum_{k=0}^\infty 2^{-4n} \1_{\{\rho^n_k \le t\}} \le \int_0^t F^{\otimes 2}_s \d [S]_s,\quad t\in [0,T].
  \end{equation}
  Plugging~\eqref{eq:model free ito cont pr3} into~\eqref{eq:model free ito cont pr1}, we get for typical price paths on the set 
  \begin{equation*}
    \{\lVert (F\cdot S)\rVert_\infty \ge a \sqrt{b} \} \cap \left\{ \int_0^T F_s^{\otimes 2} \d [ S]_s \leq b \right \}
  \end{equation*}
  that
  \begin{equation*}
    \liminf_{n \to \infty} \sup_{t \in [0,T]} \frac{\mathcal{E}^{\lambda,n}_t + \mathcal{E}^{-\lambda,n}_t}{2} \ge \frac{1}{2} \exp\bigg( \lambda a \sqrt{b} - \frac{\lambda^2}{2} b \bigg).
  \end{equation*}
  Taking $\lambda = a/\sqrt{b}$, the right hand side becomes $1/2 \exp(a^2/2)$. Our claim then follows from \cite[Remark~2.2]{Perkowski2016} which states that it suffices to superhedge with the time-supremum rather than at the terminal time.
\end{proof}

\begin{remark}
  No part of the proof was based on the fact that we are in a finite-dimensional setting, and the same arguments extend without problems to the case where we have a countable number of assets $(S^i)_{i \in \N}$ or even an uncountable number $(S^i)_{i \in I}$ but an integrand~$F$ such that $\int_0^\cdot F^{\otimes 2}_s \d [S]_s$ is well-defined and thus, in particular, with $F^i \neq 0$ only for countably many~$i \in I$.
\end{remark}

Our next aim is to extend the stochastic integral from step functions to more general integrands. For that purpose we define
\begin{equation*}
  H^2 := \left\{ F \colon \Omega_c \times [0,T] \to \R^d\,:\, \int_0^T F^{\otimes 2}_s \d [S]_s < \infty \text{ for typical price paths} \right\},
\end{equation*}
we identify $F,G \in H^2$ if $\int_0^T (F_s - G_s)^{\otimes 2} \d [S]_s = 0$ for typical price paths, and we write $\overline H^2$ for the space of equivalence classes, which we equip with the distance
\begin{equation}\label{eq:distance continuous 1}
  d_{\mathrm{QV}}(F,G) := \overline E\left[ \bigg(\int_0^T (F_t-G_t)^{\otimes 2}\d[S]_t\bigg)^{1/2} \wedge 1\right].
\end{equation}
Arguing as in~\cite[Lemma~2.11]{Perkowski2016} it is straightforward to show that $(\overline H^2,d_{\mathrm{QV}})$ is a complete metric space. If now $F$ and $G$ are step functions, for any $\varepsilon, \delta >0$ we obtain from Lemma~\ref{lem:model free ito cont} the following estimate
\begin{align*}
  d_\infty((F\cdot S),(G\cdot S)) 
  &\leq \oP ( \|((F-G)\cdot S)\|_\infty \geq \varepsilon) + \varepsilon \\
  &\leq \oP \bigg (\big \{ \|((F-G)\cdot S)\|_\infty \geq \varepsilon\big \} \cap \bigg \{ \int_0^T (F_t-G_t)^2 \d [S]_t \leq \delta \bigg \}\bigg ) \\
  &\quad+\oP \bigg (\int_0^T (F_t-G_t)^2 \d  [S]_t >\delta \bigg ) + \varepsilon \\
  &\leq 2 \exp \bigg (- \frac{\varepsilon^2}{2\delta }\bigg )+ \frac{d_{\mathrm{QV}}(F,G)}{\delta^{1/2} } + \varepsilon .
\end{align*}
Hence, setting $\delta:= d_{\mathrm{QV}}(F,G)$ and $\varepsilon:= \sqrt{\delta |\log \delta|}$ we get
\begin{equation}\label{eq:continuity continuous 1}
  d_\infty((F\cdot S),(H\cdot S)) \le 2 \delta^{1/2} + d_{\mathrm{QV}}(F,G)^{1/2} + \sqrt{\delta |\log \delta|} \lesssim d_{\mathrm{QV}}(F,G)^{1/2-\epsilon}, 
\end{equation}
for every $\epsilon \in (0,1/2)$.

\begin{remark}
  Thanks to the continuity estimate~\eqref{eq:continuity continuous 1}, we can extend the stochastic integral to the closure of the step functions in $(\overline H^2, d_{\mathrm{QV}})$. 
  
  While the space $(\overline H^2, d_{\mathrm{QV}})$ looks very much like a probability-free version of the space of (predictable) integrands feasible for classical It\^o integration, it seems to be hard to understand what kind of processes it contains.
\end{remark}

Since we do not understand this closure very well, we introduce a localized version of $d_{\mathrm{QV}}$, as in~\cite{Perkowski2016}: For $b > 0$ we define
\begin{equation*}
  d_{\mathrm{QV},b}(F,G) := \overline E\bigg[\bigg(\int_0^T (F_t-G_t)^{\otimes 2}\d[S]_t\bigg)^{1/2} \wedge\1_{\{ |[S]|_T \leq b \} } \bigg],
\end{equation*}
where we wrote $|[S]_T| := \big(\sum_{i,j=1}^d [S^i,S^j]_T^2\big)^{1/2}$, and 
\begin{equation*}
  d_{\infty,b}(F,G) :=  \overline E\big [\|F-G\|_\infty \wedge \1_{\{|[S]|_T \le b \}}\big ].
\end{equation*}
Then the same computation as before shows that
\begin{equation*}
  d_{\infty,b}((F\cdot S),(H\cdot S)) \lesssim d_{\mathrm{QV},b}(F,G)^{1/2-\epsilon}
\end{equation*}
for all step functions~$F$ and~$G$ and all $b > 0$ and every $\epsilon \in (0,1/2)$. Setting 
\begin{equation}\label{eq:distance continuous 2}
  d_{\infty,loc}(F,G):= \sum_{n=1}^{\infty} 2^{-n}d_{\infty,2^n}(F,G) \quad \text{and}\quad
  d_{\mathrm{QV},loc}(F,G) := \sum_{n=1}^{\infty} 2^{-n} d_{\mathrm{QV},2^n}(F,G)
\end{equation}
for $n\in \N$, we arrive at the estimates
\begin{align}\label{eq:continuity continuous 2}
  \begin{split}
  d_{\infty,loc} ((F\cdot S),(H\cdot S)) 
  &\lesssim \sum_{n=1}^{\infty} 2^{-n} d_{\mathrm{QV},2^n}(F,G)^{1/2-\epsilon}\\
  &\lesssim \sum_{n=1}^{\infty} 2^{-n(1/2+\epsilon)} \big(2^{-n} d_{\mathrm{QV},2^n}(F,G)\big)^{1/2-\epsilon}
  \lesssim d_{\mathrm{QV},loc}(F,G)^{1/2-\epsilon}
  \end{split}
\end{align}
for $\epsilon \in (0,1/2)$. From here the similar approximation scheme as in~\cite[Theorem~3.5]{Perkowski2016} (note that $d_{\mathrm{QV},loc}\lesssim d_{\infty,loc}$) shows that the closure of the step functions in the metric $d_{\mathrm{QV},loc}$ contains at least the left-continuous versions of c\`adl\`ag adapted processes.

\subsection{Integration for c\`adl\`ag paths with jumps restricted downwards}\label{subsec:integration cadlag}

The construction of model-free It\^o integrals with respect to c\`adl\`ag price paths requires different techniques compared to those used  in Subsection~\ref{subsec:integration continuous} for continuous price paths. While there, using the Lebesgue stopping times, we had a very precise control of the fluctuations of continuous price paths, this is not possible anymore in the presence of jumps as now price paths could have a ``big'' jump at any time. In particular, this prevents us from applying Vovk's pathwise Hoeffding inequality.\smallskip

Based on the pathwise Burkholder-Davis-Gundy inequality due to Beiglb\"ock and Siorpaes~\cite{Beiglbock2015}, we obtain the following model-free bound on the magnitude of the pathwise stochastic integral.

\begin{lemma}[Integral estimate, c\`adl\`ag version]\label{lem: ito inequality cadlag}
  For a step function $F\colon \Omega_{\psi}\times [0,T] \to \R^d$ and for $a,b,c,M > 0$ one has
  \begin{align*}
    \oQ\bigg(\big\{\lVert (F\cdot S)\rVert_\infty \geq a \big\}\cap\bigg\{\int_0^T F^{\otimes 2}_t \d [S]_t \leq b \bigg\}\cap \big\{ \|F\|_{\infty}\leq c\big\} \cap &\big\{ \|S\|_{\infty}\leq M\big\}\bigg)\\
    &\leq \frac{6\sqrt{b}+2c+2cM}{a},
  \end{align*}
  where $\oQ$ denotes the set function of Definition~\ref{def:set function}.
\end{lemma}

As it is one of the main ingredients in the proof, let us briefly recall the pathwise version of the Burkholder-Davis-Gundy inequality~\cite[Theorem~2.1]{Beiglbock2015}:
If $n\in \N,$ $k=0,\dots, n$, $x_k\in \R$, 
\begin{equation*}
  x_k^*:= \max_{0 \leq l \leq k} |x_l|\quad \text{and}\quad [x]_k :=|x_0|^2 + \sum_{l=0}^{k-1}|x_{l+1}-x_l|^2,
\end{equation*}
then 
\begin{equation}\label{eq:bdq inequality}
  x_n^*\leq 6  \sqrt{[x]_n}+2 (h\cdot x)_n,
\end{equation}
where
\begin{equation}\label{eq:h}
  (h\cdot x)_n := \sum_{k=0}^{n-1} h_k (x_{k+1}-x_k)\quad \text{with}\quad  h_k := \frac{x_k}{\sqrt{[x]_k + (x^*_k)^2}}
\end{equation}
and with the convention $\frac{0}{0}=0$. With this purely deterministic inequality at hand we are ready to prove Lemma~\ref{lem: ito inequality cadlag}.

\begin{proof}[Proof of Lemma~\ref{lem: ito inequality cadlag}]
  Let $F\colon \Omega_{\psi}\times [0,T] \to \R^d$ be a step function of the form~\eqref{eq:step function}, i.e.
  \begin{equation*}
    F_t := F_0 \1_{\{0\}}(t)+\sum_{i=0}^{\infty} F_{\sigma_i} \1_{(\sigma_i,\sigma_{i+1}]}(t),\quad t\in [0,T],
  \end{equation*}
  for some sequence of stopping times $(\sigma_i)_{i\in \N}$. For $n\in \N$ we recall that $(\tau^n_j)_{j\in \N}$ is the sequence of Lebesgue stopping times defined in Definition~\ref{def:multi Lebesgue partition} and denote by $(\rho_k^n)_{k\in\N}$ the union of $(\sigma_i)_{i\in \N}$ and $(\tau^n_j)_{j\in \N}$ with redundancies deleted. It is straightforward to see that
  \begin{equation*}
    F_t=F_t^n := F_0 \1_{\{0\}}(t)+ \sum_{k=0}^{\infty} F_{\rho^n_k}\1_{(\rho^n_k,\rho_{k+1}^n]}(t), \quad t\in [0,T],
  \end{equation*}
  and thus
  \begin{equation*}
    (F \cdot S)_t  = \sum_{i=0}^{\infty} F_{\sigma_i} S_{\sigma_i\wedge t,\sigma_{i+1}\wedge t}
    = \sum_{k=0}^{\infty} F_{\rho^n_k} S_{\rho^n_k\wedge t,\rho^n_{k+1}\wedge t} =  (F^n \cdot S)_t,\quad t\in [0,T].
  \end{equation*}
  In order to apply the pathwise Burkholder-Davis-Gundy inequality~\eqref{eq:bdq inequality}, we define iteratively $x_0^n:=0$ and 
  \begin{equation*}
    x^n_{k+1}:= x^n_{k}+F_{\rho_k^n} S_{\rho^n_{k}\wedge T,\rho^n_{k+1}\wedge T},\quad k\in \N,
  \end{equation*}
  and therefore~\eqref{eq:bdq inequality} yields
  \begin{equation*}
    \sup_{k\in \N} |x^n_{k}|= \sup_{k\in \N} | (F^n \cdot S)_{\rho^n_k \wedge T}|
    \leq 6 \bigg ( \sum_{k=0}^{\infty} (F_{ \rho^{n}_k} S_{\rho^n_k\wedge T,\rho^n_{k+1}\wedge T})^2 \bigg )^{\frac{1}{2}}
    + 2(h^n \cdot x)_T.
  \end{equation*}
  Due to the definition of the Lebesgue stopping times $(\tau^n_j)_{j\in \N}$, this leads to the continuous time estimate 
  \begin{equation*}
    \sup_{t \in [0,T]} |(F^n \cdot S)_t| 
    \leq 6 \bigg ( \sum_{k=0}^{\infty} \sum_{i,j=1}^{d}F^i_{\rho^{n}_k}F^j_{\rho^{n}_k} S^i_{\rho^n_k\wedge T,\rho^n_{k+1}\wedge T}S^j_{\rho^n_k\wedge T,\rho^n_{k+1}\wedge T} \bigg )^{\frac{1}{2}}
    + 2(\phi^n\cdot S)_T+ \|F\|_\infty 2^{-n}
  \end{equation*}
  where $F=(F^1,\dots,F^d)$ and $\phi^n$ is the adapted simple trading strategy given by the position $\phi^n_k:=h_k^n F_{\rho^n_k}$ with $h^n_k$ defined as in~\eqref{eq:h}. To turn $\phi^n$ into a weakly admissible strategy, we introduce the stopping time
  \begin{align*}
    \vartheta^n 
    :=& \inf\bigg\{t\geq 0 \,:\, \sum_{k=0}^{\infty} \sum_{i,j=1}^{d}F^i_{\rho^{n}_k}F^j_{\rho^{n}_k} S^i_{\rho^n_k\wedge t,\rho^n_{k+1}\wedge t}S^j_{\rho^n_k\wedge t,\rho^n_{k+1}\wedge t} \geq  b\bigg\}\\
    &\quad\wedge\inf \{t\geq 0\,:\, |F_t|\geq c\}\wedge\inf \{t\geq 0\,:\, |S_t|\geq M\}\wedge T   
  \end{align*}
  for $n\in \N$. Thus, we have
  \begin{align*}
    \sup_{t \in [0,\vartheta^n]}& |(F^n \cdot S)_t|\\
    &\leq 6 \bigg ( \sum_{k=0}^{\infty} \sum_{i,j=1}^{d}F^i_{\rho^{n}_k}F^j_{\rho^{n}_k} S^i_{\rho^n_k\wedge\vartheta^n,\rho^n_{k+1}\wedge \vartheta^n}S^j_{\rho^n_k\wedge\vartheta^n,\rho^n_{k+1}\wedge\vartheta^n} \bigg )^{\frac{1}{2}}
    + 2(\1_{[0,\vartheta^n]}\phi^n\cdot S)_{T}+ c 2^{-n}
  \end{align*}
  and in particular $\1_{[0,\vartheta^n]} 2 \phi^n$ is weakly $(6\sqrt{b} + 2^{-n}c + 2c+2cM)$-admissible since $|\phi^n|\leq \|F\|_{\infty}$ for every $n\in \N$.
  Taking the limit inferior as $n\to \infty$, one gets 
  \begin{align*}
    \liminf_{n\to \infty}\sup_{t \in [0,\vartheta^n]} |(F \cdot S)_t|
    \leq 6 \bigg (\int_0^T F^{\otimes 2}_t \d [S]_t \bigg)^{\frac{1}{2}}
    + 2\liminf_{n\to \infty}(\1_{[0,\vartheta^n]}\phi^n\cdot S)_{T}.
  \end{align*}
  Hence, we deduce
  \begin{equation*}
    \oQ\bigg(\big\{\lVert (F\cdot S)\rVert_\infty \geq a \big\}\cap\bigg\{\int_0^T F^{\otimes 2}_t \d [S]_t \leq b \bigg\}\cap \big\{ \|F\|_{\infty}\leq c\big\} \cap \big\{ \|S\|_{\infty}\leq M\big\} \bigg) \leq \frac{6\sqrt{b}+2c+2cM}{a}.
  \end{equation*}
\end{proof}

\begin{corollary}\label{cor:cadlag concentration estimate}
  For $a,b,c,M > 0$ and any step function $F\colon \Omega_{\psi} \times [0,T] \to \R^d$ one has
  \begin{align*}
    \oP\big(\big\{\|(F\cdot S)\|_\infty \geq a \big\} \cap\big\{\|F\|_\infty \leq c\big\}\cap \big\{ |[S]_T|\leq b \big\}
    &\cap\big \{ \Vert S \Vert_{\infty} \leq M \big\} \big) \\
    &\leq  (1+ 3 d M +2 d \psi (M))\frac{6\sqrt{b}+2+2M}{a}c,
  \end{align*}
  where we recall that $|[S]_T| := \big(\sum_{i,j=1}^d [S^i,S^j]_T^2\big)^{1/2}$.
\end{corollary}

\begin{proof}
  Using the monotonicity of $\oP$, the Cauchy-Schwarz inequality and Lemma~\ref{lem:outer measure equivalence}, we get
  \begin{align*}
    \oP\big(&\big\{\|(F\cdot S)\|_\infty \geq a \big\}\cap\big\{\|F\|_\infty \leq c \big\}\cap \big\{ |[S]_T| \leq b \big\}\cap \big\{ \Vert S \Vert_{\infty} \leq M \big\} \big)\\
    &\leq \oP\bigg(\big\{\lVert (F\cdot S)\rVert_\infty \geq a \big\}\cap\bigg\{\int_0^T F^{\otimes 2}_t \d [S]_t \leq b c^2 \bigg\}\cap \big\{ \|F\|_{\infty}\leq c\big\} \cap \big\{ \Vert S \Vert_{\infty} \leq M \big\}\bigg)\\
    &\leq (1+ 3 d M +2 d \psi (M))\\
    &\qquad\times\oQ\bigg(\big\{\lVert (F\cdot S)\rVert_\infty \geq a \big\}\cap\bigg\{\int_0^T F^{\otimes 2}_t \d [S]_t \leq b c^2 \bigg\}\cap \big\{ \|F\|_{\infty}\leq c \big\} \cap \big\{ \Vert S \Vert_{\infty} \leq M \big\} \bigg).
  \end{align*}
  Combinig this estimate with Lemma~\ref{lem: ito inequality cadlag} we get the assertion.
\end{proof}

Similarly as before we introduce the (pseudo-)distance $d_{\infty, b, M}$ on the space (of equivalence classes) of adapted processes from $\Omega_{\psi} \times [0,T]$ to $\R^d$, which is given by
\begin{equation*}
  d_{\infty, b,M}(X,Y):=\overline{E}[\Vert X-Y\Vert _{\infty}\wedge {\bf 1}_{\Omega_{b,M}}]
\end{equation*}
for $b,M>0$ and
\begin{equation*}
  \Omega_{b,M} :=\big \{ |[S] _{T}|\leq b\text{ and }\Vert S\Vert _{\infty}\leq M\big\}.
\end{equation*}
From Corollary~\ref{cor:cadlag concentration estimate} we get 
\begin{align*}
  d_{\infty, b,M}((F\cdot S),(G\cdot S)) 
  & \leq \overline{P} (\{ \Vert ((F-G)\cdot S)\Vert _{\infty}\geq a \} \cap \Omega_{b,M} )+a\\
  & \leq \overline{P} (\{ \Vert ((F-G)\cdot S)\Vert _{\infty}\geq a \} \cap \{\Vert F-G\Vert_{\infty}\geq c \} \cap \Omega_{b,M} )\\
  & \quad+\overline{P} (\{ \Vert ((F-G)\cdot S)\geq a \Vert _{\infty}\} \cap \{\Vert F-G\Vert_{\infty} \leq c \} \cap \Omega_{b,M} )+a\\
  & \leq  (1+ 3 d M +2 d \psi (M))\frac{6\sqrt{b}+2+2M}{a}c + \frac{d_{\infty, b,M}(F,G)}{c} +a
\end{align*}
for step functions $F,G$ and $a,b,c,M>0$. Setting 
\begin{equation*}
  a:=d_{\infty,b,M}(F,G)^{1/3}\quad \text{and}\quad c:= d_{\infty,b,M}(F,G)^{2/3},
\end{equation*}
we deduce that 
\begin{equation*}
  d_{\infty,b,M}((F\cdot S),(G\cdot S)) \leq  (1+ 3 d M +2 d \psi (M)) (6\sqrt{b}+4+2M) d_{\infty,b,M}(F,G)^{1/3}.
\end{equation*}
Defining for some fixed $\epsilon\in (0,1)$ the metric
\begin{equation}\label{eq:distance cadlag}
  d_{\infty,\psi}(F,G):= \sum_{n,m=1}^\infty  2^{-(n/2+m)(1+\epsilon)} (\psi (2^m)\vee2^m\vee 1)^{-1} (d_{\infty,2^n,2^m}(F,G)\wedge 1)
\end{equation}
we obtain 
\begin{equation}\label{eq:continuity cadlag}
   d_{\infty,\psi}((F\cdot S),(G\cdot S)) 
    \lesssim d_{\infty}(F,G)^{1/3}.
\end{equation}
Based on this observation, we can again extend the construction of the model-free It\^o integral from simple integrands to more general integrands with respect to typical c\`adl\`ag price paths with jumps restricted downwards.

\appendix 
\section{Properties of Vovk's outer measure}\label{sec:appendix}

This appendix collects postponed proofs from the previous sections and an elementary result (Borel-Cantelli lemma), which was used for the construction of the quadratic variation and the model-free It\^o integrals.

\begin{proof}[Proof of Proposition~\ref{prop:local martingale}]
  Let $\lambda > 0$ and let $(H^n)_{n \in \mathbb{N}} \subseteq \mathcal{H}_\lambda$ be such that $ \liminf_n(\lambda + (H^n \cdot S)_T) \ge \1_A$. Then, we estimate
  \begin{equation*}
    \P(A) \le \E_\P[\liminf_n (\lambda + (H^n \cdot S)_T)] \le \liminf_n \E_\P[\lambda + (H^n \cdot S)_T] \le \lambda,
  \end{equation*}
 where in the last step we used that $\lambda + (H^n \cdot S)$ is a non-negative c\`adl\`ag $\P$-local martingale with $\E_{\P} [|\lambda + (H^n \cdot S)_0|]<\infty$ and thus a $\P$-supermartingale.
\end{proof}

\begin{proof}[Proof of Proposition~\ref{prop:na1 interpretation}]
  If $\overline{P}(A) = 0$, then for every $n\in \N$ there exists a sequence of simple strategies $(H^{n,m})_{m \in \N} \subset \mathcal{H}_{2^{-n-1}}$ 
  such that $2^{-n-1} + \liminf_{m \to \infty} (H^{n,m}\cdot S )_T (\omega) \geq \1_A(\omega)$ for all $\omega \in \Omega_\psi$. For $K\in (0,\infty)$ set $G^m := K \sum_{n=0}^m H^{n,m}$, 
  and thus $G^m \in \mathcal{H}_K$. For every $k \in \N$ one gets
  \begin{equation*}
    \liminf_{m \to \infty} (K + (G^m \cdot S)_T)  \ge \sum_{n=0}^k (2^{-n-1}K  + \liminf_{m \to \infty} (H^{n,m} \cdot S)_T) \ge (k+1)K \1_A.
  \end{equation*}
  Because the left hand side does not depend on~$k$, the sequence $(G^m)$ satisfies~\eqref{eq:NA1 with simple strategies}.
  
  Conversely, if there exist a constant $K\in (0,\infty)$ and a sequence of strongly $K$-admissible simple strategies $(H^n)_{n\in \N} \subset \mathcal{H}_K$ satisfying \eqref{eq:NA1 with simple strategies}, then we can scale it down by an arbitrary factor $\varepsilon > 0$ to obtain a sequence of strategies in $\mathcal{H}_\varepsilon$ that superhedge $\1_A$, which implies $\overline{P}(A) = 0$.
\end{proof}

As the proof of the classical Borel-Cantelli lemma requires only countable subadditivity, Vovk's outer measure allows for a version of the Borel-Cantelli lemma. 

\begin{lemma}\label{lem:Borel-Cantelli}
  Let $(A_j)_{j \in \N}\subseteq \Omega_\psi$ be a sequence of events. If $\sum_{j=1}^{\infty} \overline{P}(A_j)<\infty$, then 
  \begin{equation*}
    \overline{P} \bigg(\bigcap_{i=1}^\infty \bigcup_{j=i}^\infty A_j\bigg)\leq \liminf_{i \to \infty} \overline{P} \bigg(\bigcup_{j=i}^\infty A_j\bigg) \leq \liminf_{i \to \infty} \sum_{j=i}^\infty\overline{P} (A_j) = 0.
  \end{equation*}
\end{lemma}

\section{Stopping times}\label{app:stopping}

The following result is standard, but since we did not find a reference we include the proof.

\begin{lemma}\label{lem:stopping}
  Let $(\Omega, (\F^\circ_t)_{t \ge 0})$ be a filtered measurable space and let $\F_t$ be the universal completion of $\F^\circ_t$ for $t \ge 0$. Let $(X_t)_{t \ge 0}$ be an $\R^d$-valued right-continuous $(\F_t)$-adapted process, let $\tau$ be a $(\F_t)$-stopping time,  let $Y$ be an $\R^d$-valued $\F_\tau$-measurable random variable, and let $K \subset \R^d$ be a closed set. Then
  \begin{equation*}
    \rho := \inf\{ t\ge \tau: X_t + Y \in K\}
  \end{equation*}
  is a $(\F_t)$-stopping time.
\end{lemma}

\begin{proof}
  If $K = \R^d$, then $\rho = \tau$ is a stopping time. So assume $K \subsetneq \R^d$, let $x \in \R^d \setminus K$, and define the auxiliary process $Z_t := \1_{\{t < \tau\}} x + \1_{\{t \ge \tau\}} (X_t + Y)$. Then $Z$ is right-continuous and adapted and therefore progressively measurable, and $\rho = \inf\{t \ge 0: Z_t \in K\}$. Define for $t \ge 0$ the set
  \begin{equation*}
    A_t := \{(s,\omega) \in [0,t] \times \Omega: Z_s(\omega) \in K\} \in \mathcal B([0,t]) \otimes \F_t.
  \end{equation*}
  By \cite[Theorem Appendix.III.82]{Dellacherie1978} there exists a projection $\Pi \colon [0,t] \otimes \Omega \to \Omega$ such that $\Pi(A) = \{\omega \in \Omega: \exists s \in [0,t], (s,\omega) \in A\} \in \F_t$ for all $A \in \mathcal B([0,t]) \otimes \F_t$ (it is here that we use that $\F_t$ is universally completed). Therefore, using the right-continuity of $Z$ and the closedness of $K$,
  \begin{equation*}
    \{\omega: \rho(\omega) \le t\} = \{\omega: \exists s \in [0,t]: Z_s(\omega) \in K\} = \{\omega: \exists s \in [0,t]: (s,\omega) \in A_t \} = \Pi(A_t) \in \F_t,
  \end{equation*}
  and thus $\rho$ is a stopping time.
\end{proof}

\bibliography{quellen}{}
\bibliographystyle{amsalpha}

\end{document}